\documentclass[11pt]{article}

\usepackage{fullpage}
\usepackage{graphicx}
\usepackage{xcolor}
\usepackage{amssymb}
\usepackage{amsmath}
\usepackage{amsthm}
\usepackage{mathpazo}

\usepackage{thmtools}
\usepackage{thm-restate}

\usepackage{algorithm}
\usepackage{algpseudocode}

\usepackage[hidelinks]{hyperref}
\hypersetup{
    colorlinks,
    linkcolor={red!80!black},
    citecolor={blue!80!black},
    urlcolor={blue!80!black}
}
\usepackage{cleveref}

\algtext*{EndWhile}
\algtext*{EndFor}
\algtext*{EndIf}

\mathchardef\mhyphen="2D
\newcommand{\codt}{\mathsf{cstr\mhyphen opt\mhyphen dcmp\mhyphen tree}}
\newcommand{\clt}{\mathsf{cstr\mhyphen label\mhyphen tree}}
\newcommand{\colcst}{\mathsf{cstr\mhyphen opt\mhyphen LCST}}

\newcommand{\R}{\mathbb{R}}

\newtheorem{theorem}{Theorem}[section]
\newtheorem{lemma}[theorem]{Lemma}
\newtheorem{claim}[theorem]{Claim}
\newtheorem{coro}[theorem]{Corollary}

\newtheorem{hypothesis}[theorem]{Hypothesis}

\theoremstyle{definition}
\newtheorem{definition}[theorem]{Definition}

\theoremstyle{remark}
\newtheorem{remark}[theorem]{Remark}

\newcommand{\bfT}{{\mathbf{T}}}
\newcommand{\bfH}{{\mathbf{H}}}
\renewcommand{\root}{{\mathrm{root}}}

\newcommand{\dem}{{\mathrm{dem}}}
\newcommand{\ser}{{\mathrm{ser}}}
\newcommand{\head}{{\mathrm{head}}}
\newcommand{\tail}{{\mathrm{tail}}}

\newcommand{\ceil}[1]{\left\lceil#1\right\rceil}

\newcommand{\poly}{\mathrm{poly}}
\newcommand{\polylog}{\mathrm{polylog}}

\newcommand{\opt}{\mathrm{opt}}

\newcommand{\calC}{\mathcal{C}}
\newcommand{\calE}{\mathcal{E}}
\newcommand{\calP}{\mathcal{P}}

\newcommand{\calV}{\mathcal{V}}

\newcommand{\cost}{\mathrm{cost}}

\newcommand{\edge}{\mbox{\sf Edge}}

\everymath{\displaystyle}

\usepackage{enumitem}
\setlist{itemsep=3pt,parsep=0pt,topsep=3pt}

\newcommand{\leaf}{{\mathrm{leaf}}}
\newcommand{\internal}{{\mathrm{int}}}
\newcommand{\solve}{{\mathsf{solve}}}

\DeclareMathOperator*{\E}{\mathbb{E}}

\newcommand{\old}{{\mathrm{old}}}
\newcommand{\new}{{\mathrm{new}}}

\newcommand{\bbD}{{\mathbb{D}}}

\newcommand{\SA}{{\mathrm{SA}}}

\newcommand{\calX}{{\mathcal{X}}}

\graphicspath{{figs/}}

\newcommand{\eps}{{\varepsilon}}

\ifdefined\DEBUG
 \newcommand{\fab}[1]{\textcolor{orange}{#1}}
 \newcommand{\sh}[1]{\textcolor{red}{#1}}
  \newcommand{\bu}[1]{\textcolor{blue}{#1}}
  
  \def\rem#1{{\marginpar{\raggedright\scriptsize #1}}}
  \newcommand{\fabr}[1]{\rem{\small\textcolor{orange}{$\bullet${\tiny #1}}}}
  \newcommand{\shr}[1]{\rem{\small\textcolor{red}{$\bullet${\tiny #1}}}}
  \newcommand{\bur}[1]{\rem{\textcolor{blue}{$\bullet${\tiny #1}}}}
\else
  \newcommand{\fab}[1]{#1}
  \newcommand{\fabr}[1]{}
  \newcommand{\sh}[1]{#1}
  \newcommand{\shr}[1]{}
  \newcommand{\bu}[1]{#1}
  \newcommand{\bur}[1]{}
\fi

\begin{document}

\title{$O(\log^2k/\log\log{k})$-Approximation Algorithm for Directed Steiner Tree: A Tight Quasi-Polynomial-Time Algorithm.}

\author{
	Fabrizio Grandoni \thanks{IDSIA, USI-SUPSI, E-mail: \url{fabrizio@idsia.ch}.}
	\and
	Bundit Laekhanukit \thanks{Institute for Theoretical Computer Science, Shanghai University of Finance and Economics. E-mail: \url{bundit@sufe.edu.cn}.} 
	\and
	Shi Li \thanks{Department of Computer Science and Engineering, University at Buffalo. E-mail: \url{shil@buffalo.edu}.}
	}

\date{\today}

\maketitle

\pagenumbering{gobble}

%{
%\color{red} TODO list 
%\begin{itemize}
%	\item Keep $O(\log^2k)$ factor leaves $O(\log^2k/\log\log k)$ for later.
%    \item Use "node" or "vertex"? Use vertex of vertices in $G$, and use nodes for nodes in decomposition trees and labelled tree. 
%\end{itemize}
%}

\begin{abstract}
\noindent In the Directed Steiner Tree (DST) problem we are given an $n$-vertex directed edge-weighted graph, a root $r$, and a collection of $k$ terminal nodes. Our goal is to find a minimum-cost arborescence that contains a directed path from $r$ to every terminal. We present an $O(\log^2 k/\log\log{k})$-approximation algorithm for DST that runs in quasi-polynomial-time, i.e., in time $n^{\poly\log \fab{(}k)}$.
By assuming the Projection Game Conjecture and $\mathrm{NP}\not\subseteq{\bigcap}_{0<\epsilon<1}\mathrm{ZPTIME}(2^{n^\epsilon})$, and adjusting the parameters in the hardness result of Hal\fab{p}erin and Krauthgamer [STOC'03], we show the matching lower bound of $\fab{\Omega}(\log^2{k}/\log\log{k})$\fabr{it was $\omega$} for the class of quasi-polynomial-time algorithms, meaning that our approximation ratio is asymptotically the best possible.
This is the first improvement on the DST problem since the classical quasi-polynomial-time $O(\log^3 k)$ approximation algorithm by Charikar~et~al. [SODA'98 \& J.~Algorithms'99]. (The paper erroneously claims an $O(\log^2k)$ approximation due to a mistake in prior work.) 
	
Our approach is based on two main ingredients. First, we derive an approximation preserving reduction to the {\em Label-Consistent Subtree (LCST)} problem. Here we are given a rooted tree with node labels, and a feasible solution is a subtree satisfying proper constraints on the labels. The LCST instance has quasi-polynomial size and logarithmic height. We remark that, in contrast, Zelikovsky's heigh-reduction theorem [Algorithmica'97] used in all prior work on DST achieves a reduction to a tree instance of the related Group Steiner Tree (GST) problem of similar height, however losing a logarithmic factor in the approximation ratio.

%we derive a height-reduction technique that does not pay any loss in the approximation factor; this improves the seminal result of Zelikovsky [Algorithmica'97] that has been used in all the literature for DST (e.g., [Rothvo{\ss}, Preprint'11; Friggstad~et~al., IPCO'14]) and related problems [C{\u{a}}linescu-Zelikovsky, J. Comb. Optim.; Chakrabarty~et~al., FOCS'15, Grandoni-Laekhanukit, STOC'17].  Our height-reduction theorem, unlike the original one that reduces the DST problem to the Group Steiner Tree (GST) problem on trees, reduces DST to what  we call the {\em Label-Consistent Subtree (LCST)} problem.\fab{Roughly speaking, in this problem we are given a rooted tree with node labels of two types. Global labels are used to enforce that all terminals are connected to the root, while local labels are used to guarantee that the computed solution induces an arborescence rooted at $r$ in the original problem.}

% this problem introduces an additional {\em label-consistent} constraint that could not be handled by standard DST algorithms.

Our second ingredient is an LP-rounding algorithm to approximately solve LCST instances, which is inspired by the framework developed by [Rothvo{\ss}, Preprint'11; Friggstad~et~al., IPCO'14]. We consider a Sherali-Adams lifting of a proper LP relaxation of LCST. Our rounding algorithm proceeds level by level from the root to the leaves, \fab{rounding and} conditioning each time on a proper subset of \emph{label} variables. The limited height of the tree and small number of labels on root-to-leaf paths guarantees that a small enough (namely, polylogarithmic) number of Sherali-Adams lifting levels is sufficient to condition up to the leaves.

\fab{We believe that our basic strategy of combining label-based reductions with a round-and -condition type of LP-rounding over hierarchies might find applications to other related problems.}\fabr{Tried to add a sentence, check}

%We believe that the combination of node-labeled reductions and 

%Second, to handle the label-consistent constraint, we apply Sherali-Adam\fab{s} hierarchy to derive a strong LP that allows us to round a fractional solution to the LCST instance efficiently (in quasi-polynomial-time), losing a factor of $O(\log^2 k/\log\log{k})$.
%%We apply LP-Hierarchy techniques to ensure the label-consistency requirement, that 
%This is based on the framework developed by Rothvo{\ss} (which was later simplified by Friggstad~et~al.).

\end{abstract}

\newpage
\pagenumbering{arabic}

\setcounter{page}{1}

\section{Introduction}
\label{sec:intro}

In the {\em Directed Steiner Tree  (DST)} problem,
we are given an $n$-vertex digraph $G=(V,E)$ with cost $c_e$ on each edge
$e\in E$, a root vertex $r\in V$ and a set of $k$ terminals 
$K \subseteq V\setminus \{r\}$.
The goal is to find a minimum-cost out-arborescence $H\subseteq G$ rooted at $r$ that contains an $r\to t$ directed path for every terminal $t \in K$. W.l.o.g. we assume that edge costs satisfy triangle inequality.
%\emph{Throughout the paper, we assume that $G$ is already a metric completion of itself.} That is, the costs of edges satisfy triangle inequalities.  This can be assumed w.l.o.g and used to improve $O(\log n)$-factors to $O(\log k)$.

%W.l.o.g we can assume that $H$ is out-arborescence rooted at $r$, hence the name.
The DST problem is a fundamental problem in the area of network design that is known for its bizarre behaviors. While constant-approximation algorithms have been known for its undirected counterpart (see, e.g., \cite{ByrkaGRS13,RobinsZ05,Zelikovsky93}),
the best known polynomial-time approximation algorithm for this problem could achieve only an $O((1/\epsilon)^3k^{\epsilon})$ approximation ratio
in time $O(n^{1/\epsilon})$ for any $0 < \epsilon \leq 1/\log_2{k}$, due to the classical work of Charikar~et~al.~\cite{CharikarCCDGGL99}.
Even allowing \fab{this} algorithm to run in quasi-polynomial-time, the best approximation ratio remains $O(\log^3k)$~\cite{CharikarCCDGGL99}\footnote{The original paper claims an $O(\log^2 k)$-approximation algorithm; however, their result was based on the initial statement of the Zelikovsky's height-reduction theorem in \cite{Zelikovsky97}, which was later found to contain a subtle flaw and was restated by Helvig, Robin and Zelikovsky \cite{HelvigRZ01}.}.
Since then, there have been efforts to get improvements either in the running-time or in the approximation guarantee of this problem, e.g, using the the primal-dual method~\cite{ZosinK02}, Sum-of-Square\fab{s} (a.k.a. Lasserre) hierarchy \cite{Rothvoss11}, Sherali-Adams and Lov{\'{a}}sz-Schrijver hierarchies~\cite{FriggstadKKLST14}. Despite all these efforts, there has been no significant improvement over the course of the last two decades for both polynomial and quasi-polynomial time algorithms. % for polynomial-time algorithms.
In fact, it is known from the work of Halperin and Krauthgamer \cite{HalperinK03} that unless $\mathrm{NP}\subseteq\mathrm{ZPTIME}(n^{\mathrm{polylog}(n)})$, it is not possible to achieve an approximation ratio $O(\log^{2-\epsilon}k)$, for any constant $\epsilon>0$, and such lower bound applies to both polynomial and quasi-polynomial time algorithms.
%\hg{Throughout the paper, one should think $\log k$ and $\log n$ as being of the same order.}\fabr{I don't think this sentence is clear, and also it should only refer to the lower bound I think} 
This means that there is a huge gap between the upper bound of $k^{\epsilon}$ and the lower bound of $\log^{2-\epsilon} k$ for polynomial-time algorithms. All efforts were failed to obtain even an $n^{o(1)}$-approximation algorithm that runs in polynomial-time.

For the class of quasi-polynomial-time algorithms, the approximation ratio of $O(\log^3k)$ is arguably disappointing. 
This is because its closely related special case, namely, the {\em Group Steiner Tree} (GST) problem,
%\footnote{In GST, we are given an $n$-vertex {\bf undirected} graph $G$ with edge-costs, a root vertex $r$ and a collection of subsets $S_1,\ldots,S_k$ of vertices ({\em groups}); the goal is to find a minimum-cost subgraph that contains a path from the root to at least one vertex of each group. One can reduce GST to DST by making $G$ bi-directed and then add a terminal $t_i$, for each group $S_i$, with zero cost edges directed from every vertex of $S_i$ to $t_i$.} 
is known to admit a quasi-polynomial-time $O(\log^2{k})$-approximation algorithm on general graphs due to the work of Chekuri and Pal \cite{ChekuriP05}. A natural question would be whether such an approximation ratio could be achieved in quasi-polynomial-time for DST as well. Nevertheless, achieving this improvement with the known techniques seems to be impossible. Indeed, all previous algorithms for DST \cite{CharikarCCDGGL99,Rothvoss11,FriggstadKKLST14} rely on the well-known Zelikovsky's height-reduction theorem \cite{Zelikovsky97,HelvigRZ01}. These algorithms (implicitly) reduce DST to GST on trees, which loses an $\sh{\Theta}(\log k)$ approximation factor in the process. Furthermore, the $\sh{\Omega}(\log^{2-\epsilon}k)$-hardness \shr{it was an $O$} of Halperin and Krauthgamer \cite{HalperinK03} carries over to GST on trees.  We remark that algorithms for many related problems (see, e.g., \cite{EneCKP15,GrandoniL17}) rely on the same height-reduction theorem. 

\subsection{Our Results and Techniques}

The purpose of this work is to close the gap between the lower and upper bounds on the approximability of DST in quasi-polynomial time. Our main result is as follows.
%an $O(\log^2k/\log\log{k})$-approximation algorithm for DST that runs in quasi-polynomial-time.  Our main result is 
\begin{theorem}
	\label{thm:main}
	There is a randomized $O(\log^2k/\log \log k)$-approximation algorithm for DST with running time $n^{O(\log^5 k)}$.
\end{theorem}

By analyzing the proofs in \cite{HalperinK03}, we also show that this bound is asymptotically tight under stronger assumptions; please see more discussion in \Cref{appendix:hardness-quasi-poly}.
\begin{theorem}
	\label{thm:negative} There is no quasi-polynomial-time algorithm for DST that achieves an approximation ratio $o(\log^2{k}/\log\log k)$ unless $\mathrm{NP}\subseteq{\bigcap}_{0<\epsilon<1}\mathrm{ZPTIME}({2^{n^{\epsilon}}})$ or the {\em Projection Game Conjecture} is false.
\end{theorem}

Our upper bound is based on two main ingredients. The first one is a quasi-polynomial-time approximation-preserving reduction to a novel {\em Label-Consistent Subtree} (LCST) problem. Roughly speaking, in LCST we are given a rooted tree plus node labels of two types, global and local. A feasible solution consists of a subtree that satisfies proper constraints on the labels. Intuitively, local labels are used to guarantee that a feasible solution induces an arborescence rooted at $r$ in the original problem, while global labels are used to enforce that all the terminals are included in such arborescence. In our reduction the tree has size $n^{\poly\log(k)}$ and height $\fab{h}=O(\log k/\log \log k)$,\fabr{$h$ used later} with $k$ global labels. For a comparison, Zelikovsky's height-reduction theorem \cite{Zelikovsky97}, used in all prior work on DST, reduces (implicitly) the latter problem to a GST instance over a tree of height $O(\log k)$. However, this reduction alone loses a factor $\sh{\Theta}(\log k)$ in the approximation (while our reduction is approximation-preserving).

Our second ingredient is a quasi-polynomial-time $O(\log^2 k/\log\log k)$-approximate LP-rounding algorithm for LCST instances arising from the previous reduction. Here we exploit the LP-hierarchy framework developed by Rothvo{\ss}~\cite{Rothvoss11} (and later simplified by Friggstad~et~al.~\cite{FriggstadKKLST14}). We define a proper LP relaxation for the problem, and solve an $R$-level Sherali-Adams lifting of this LP for a parameter $R=\poly\log k$. We then round the resulting fractional solution level by level from the root to the leaves. At each level we maintain a small set of labels that must be provided by the subtree. By randomly round\fab{ing} label-based variables and conditioning, we push the set of labels all the way down to the leaves, guaranteeing that the output tree is always label-consistent.
%randomly round proper \emph{label-based variables}, and then \emph{condition} on them. 
Thanks to the limited height of the tree and to the small number of labels along root-to-leaf paths, a polylogarithmic number of lifting levels is sufficient to perform the mentioned conditioning up to the leaves. As in \cite{Rothvoss11}, the probability that each global label appears in the tree we directly construct  is only $1/(h+1)$. We need to repeat the process $O(h\log k) = O(\log^2k/\log\log k)$ times in order to make sure all labels are included with high probability, leading to the claimed approximation ratio.\fabr{Maybe the readers expects a few citations here} Our result gives one more application of using LP/SDP hierarchies to obtain improved approximation algorithms, in addition to a few other ones (see, e.g., \cite{BCG09,Chlamtac07,CGM13,LR16,GKL19}). 

\fab{We believe that our basic strategy of combining a label-based reduction with a \emph{round-and-condition} rounding strategy as mentioned above might find applications to other problems, and it might therefore be of independent interest.}\fabr{Added a sentence like in abstract}

\subsection{Comparison to Previous Work}
\label{sec:comparison}
Our algorithm is inspired by two results. First is the recursive greedy algortihm of Chekuri and Pal for GST \cite{ChekuriP05}, and second is the hierrachical based LP-rounding techniques by Roth\fab{vo}{\ss} \cite{Rothvoss11}.

As mentioned, the algorithm of Chekuri and Pal is the first \sh{one} that yields an\fabr{Removed repetition almost tight} approximation ratio of $O(\log^2k)$ for GST, which is a special case of DST, in quasi-polynomial-time. This is almost tight for the class of quasi-polynomial-time algorithms. Their algorithm exploits the fact that any optimal solution can be shortcut into a path of length $k$, while paying only a factor of 2 (such path exists in the metric-closure of the input graph). This simple observation allows them to derive a recursive greedy algorithm. In more detail, they try to identify a vertex that separates the optimal path into two equal-size subpaths by iterating over all the vertices; then they recursively (and approximately) solve two subproblems and pick the best approximate sub-solution greedily.
Their analysis, however, requires the fact that both recursive calls end at the same depth (because each subpath has length different by at most one). 

%Our reduction from DST to LCST is inspired by the work of Chekuri and Pal for GST \cite{ChekuriP05}.
\fabr{Removed duplicated sentence} 
We imitate the recursive greedy algorithm by recursively splitting the optimal solution via balanced tree separators. The same approach as in \cite{ChekuriP05}, unfortunately, does not quite work out for us \fab{since subproblem sizes may differ by a multiplicative factor}. \shr{Probably Fab has the same question here: why is differing by a multiplicative factor not OK?} This process, somehow, gives us a decision tree that contains a branch-decomposition of every solution, which is sufficient to devise an approximation algorithm.
Note, however, that not every subtree of this decision tree can be transformed into a connected graph, and thus, it is not guaranteed that we can find a feasible DST solution from this decision tree. We introduce node-labels and label-consistent constraints specifically to solve this issue. 

The label-consisten\fab{cy} \sh{requirement} could \sh{not} be handled simply by applying DST algorithms as a blackbox. This comes to the second component that is inspired by the framework developed by Rothvo{\ss}~\cite{Rothvoss11}. 
\bu{While the framework was originally developed for the Sum-of-Squares hierarchy, it was shown by Friggstad~et~al.~\cite{FriggstadKKLST14} that it also applies to Sherali-Adams, which is a weaker hierarchy. 
We apply the framework of Rothvo{\ss} to our Sherali-Adams lifted-LP} \sh{but taking the label-consistency requirement into account}.
%There are two main differen\fab{ces} between our algorithm and Rothvo{\ss}'s \fab{one}. First, Rothvo{\ss} applies the Sum-of-Square\fab{s} hierarchy to the standard LP for DST, while we apply a weaker Sherali-Adams hierarchy to an LP with label-consistency constraints.
%Second, Rothvo{\ss} applies a more standard Lift-and-Project rounding algorithm, which can be phrased as applying the randomized rounding algorithm for GST on a tree by Garg~et~al.~\cite{GargKR00} to the support of the lifted-LP. \fab{Our algorithm deviates from this since it applies the randomized rounding strategy to different distributions in parallel.}

%, on the other hand, while based on the same framework, we are required to run some parts of the randomized rounding algorithms with different distribution in parallel. This part deviates us from the previous rounding algorithm.

\subsection{Related Work}
\label{sec:related}

We already mentioned some of the main results about DST and GST. For GST there is a polynomial-time algorithm by Garg~et~al.~\cite{GargKR00} that achieves an approximation factor of $O(\log^2 k\log n)$, where $k$ is the number of groups. Their algorithm first maps the input instance into a tree instance by invoking the {\em Probabilistic Metric-Tree Embeddings} \cite{Bartal96,FakcharoenpholRT04}, thus losing a factor $O(\log n)$ in the approximation ratio. They then apply an elegant LP-based randomized rounding algorithm to the instance on a tree. A well-known open problem is whether it is possible to avoid the $\log n$ factor in the approximation ratio. 
This was later achieved by Chekuri and Pal \cite{ChekuriP05}, however the\fab{ir} algorithm runs in quasi-polynomial-time.

Some works were devoted to the \emph{survivable network} variants of DST and GST, namely $\ell$-DST and $\ell$-GST, respectively. Here one requires to have $\ell$ edge-disjoint directed (resp., undirected) paths from the root to each terminal (resp., group). Cheriyan~et~al. \cite{CheriyanLNV14} showed that $\ell$-DST admits no $2^{\log^{1-\eps}n}$-approximation algorithm, for any $\eps>0$, unless $\mathrm{NP}\subseteq\mathrm{DTIME}(2^{\polylog(n)})$.
Laekhanukit \cite{Laekhanukit14} showed that the problem admits no $\fab{\ell}^{1/2-\eps}$-approximation for any constant $\eps>0$, unless $\mathrm{NP}=\mathrm{ZPP}$.
Nevertheless, the negative results do not rule out \fab{the} possibility of achieving reasonabl\fab{e} approximation factors for small values of $\ell$. In particular, Grandoni and Laekhanukit \cite{GrandoniL17} (exploiting some ideas in \cite{Laekhanukit16}) recently devised a poly-logarithmic approximation algorithm for $2$-DST that runs in quasi-polynomial time.

\fab{Concerning} $\ell$-GST, Gupta~et~al.~\cite{GuptaKR10} presented a $\tilde{O}(\log^3n\log k)$-approximation algorithm for $2$-GST. The same problem admits an $O(\alpha \log^2n)$-approximation algorithm, where $\alpha$ is the largest cardinality of a group \cite{KKN12}. Chalermsook~et~al. \cite{ChalermsookGL15} presented an LP-rounding bicriteria approximation algorithm for $\ell$-GST that returns a subgraph with cost $O(\log^2n\log k)$ times the optimum while guarantee\fab{ing} a connectivity of at least $\Omega(\ell/\log n)$. They also showed that $\ell$-GST is hard to approximate to within a factor of $\ell^{\sigma}$, for some fixed constant $\sigma>0$,
and if $\ell$ is large enough, then the problem is at least as hard as
the {\em Label-Cover} problem, meaning that $\fab{\ell}$-GST admits no
$2^{\log^{1-\eps}n}$-approximation algorithm, for any constant $\eps>0$,
unless $\mathrm{NP}\subseteq\mathrm{DTIME}(2^{\polylog(n)})$.

\section{Preliminaries}
\label{sec:prelim}
Given a graph $G'$, we denote by $V(G')$ and $E(G')$ the vertex and edge set of $G'$, respectively. 
%Given a subset of vertices $S\subseteq V(G)$, we denote by $\delta^{in}_G(S)$ and $\delta^{out}_G(S)$, the set of edges in $G$ leaving $S$. If the graph $G$ is known in  the context, we will skip the subscript.  
Throughout this paper, we treat a rooted tree as an out-arborescence; that is, edges are directed towards the leaves. Given a rooted tree $T$, we use $\root(T)$ to denote its root. For any rooted tree $T$ and $v \in V(T)$, we shall use $T[v]$ to denote the sub-tree of $T$ containing $v$ and all descendants of $v$. For a directed edge $e = (u, v)$, we use $\head(e) = u$ and $\tail(e) = v$ to denote the head and tail of $e$. 
Generally, we will use the term {\em vertex} to mean a vertex of a DST instance, and we will use the term {\em node} to mean a vertex in an instance of the Label-Consistent Subtree problem, defined below:

\paragraph{Label-Consistent Subtree (LCST).} The new problem we introduce is the Label-Consistent Subtree (LCST) problem.  The input consists of a rooted tree $T^0$ of size $N = |V(T^0)|$ and  height $h$,  a node cost vector $c \in \R_{\geq 0}^{V(T^0)}$, and a set $L$ of labels, among which there are $k$ \emph{global labels} $K \subseteq L$. The other labels $L \setminus K$ are called \emph{local labels}. Each node $v \in V(T^0)$ has two label sets: a set $\dem(v) \subseteq L \setminus K$ of \emph{demand labels}, and a set $\ser(v) \subseteq L$ of \emph{service labels}.  

We say that a subtree $T$ of $T^0$ with $\root(T) = \root(T^0)$ is {\em label-consistent} if for every vertex $u \in V(T)$ and $\ell \in \dem(u)$, there is a descendant $v$ of $u$ in $T$ such that $\ell \in \ser(v)$. % In particular, $T$ has a path joining two vertex with the same label that are in ancestor-descendant relationship. 
The goal of the LCST problem is to find a label-consistent subtree $T$ of $T^0$ of minimum cost that contains all global labels, i.e, for every $\ell \in K$, there is a $v \in V(T)$ with $\ell \in \ser(v)$.

In Section~\ref{sec:algo}, we give an $(shN)^{O(sh^2)}$-time $O(h \log k)$-approximation algorithm for the LCST problem, where $s = \max_{v \in V(T^0)}|\dem(v)|$.  \sh{Thus, we require $s$ to be small in order to derive a quasi-polynomial-time algorithm; fortunately, this is the case for the instance reduced from DST.}
%% Generally, if $s$ can be arbitrary large, LCST might be at least as hard as the Label-Cover problem;
%% please see \Cref{appendix:hardness-lcsg}.
 
One may generalize LCSs to general graphs, say {\em Label-Consistent Steiner Subgraph} (LCSS).
 %But then, it is not hard to see as this problem includes the infamous {\em Label-Cover} problem as a special case.
 %(Please see the formal definition and the hardness reduction in Appendix XXX.)
 %As such, it is sufficient to concern the case where the input graph is a tree.
 
 \paragraph{Balanced Tree Partition.}
 A main tool in our reduction is the following standard balanced-tree-partition lemma (with proof given in \Cref{apendix:balcned-tree-partition} for completeness). 
% stating that any rooted tree can be decomposed into two rooted trees of size less than $2n/3 + 1$ in such a way that they only share one vertex in common.
 \begin{restatable}[Balanced-Tree-Partition]{lemma}{balancedpartition}
  \label{cor:balanced-partition}
  For any $n \geq 3$,
  for any $n$-vertex tree $T$ rooted at a vertex $r$,
  there exists a vertex $v\in V(T)$ such that 
  $T$ can be decomposed into two trees $T_1$ and $T_2$ rooted 
  at $r$ and $v$, respectively, in such a way
  that $E(T_1)\uplus E(T_2)=E(T)$, $V(T_1) \cup V(T_2) = V(T)$ and $V(T_1)\cap V(T_2)=\{v\}$
  and $|V(T_1)|, |V(T_2)| < 2n/3+1$.
  In other words, $T_1$ and $T_2$ are sub-trees 
  that form a balanced partition of (the edges of) $T$.
  \end{restatable}
  
\paragraph{Sherali-Adams Hierarchy.} 

In this section, we give some basic facts about Sherali-Adams hierarchy that we will need.  Assume we have a linear program polytope $\calP$ defined by  $Ax \leq b$. We assume that $0\leq x_i \leq 1, \forall i \in [n]$ are part of the linear constraints. The set of integral feasible solutions is defined as $\calX = \{x \in \{0, 1\}^n: Ax \leq b\}$.  It is convenient to think of each $i \in [n]$ as an event, and in a solution $x \in \{0, 1\}^n$, $x_i$ indicates whether the event $i$ happens or not. 

The idea of Sherali-Adams hierarchy is to strengthen the original LP $Ax \leq b$ by adding more variables and constraints. Of course, each $x \in \calX$ should still be a feasible solution to the strengthened LP (when extended to a vector in the higher-dimensional space). For some $R \geq 1$, the $R$-th round of Sherali-Adams lift of the linear program has variables $x_S$, for every $S \in {[n] \choose \leq R} := \{S \subseteq [n]: |S| \leq R \}$. For every solution $x \in \calX$, $x_S$ is supposed to indicate whether all the events in $S$ happen or not in the solution $x$; that is, $x_S = {\prod}_{i \in S}x_i$. Thus each $x \in \calX$ can be naturally extended to a 0/1-vector in the higher-dimensional space defined by all the variables.

To derive the set of constraints, let us focus on the $j$-th constraint $\sum_{i=1}^na_{j,i} x_i \leq b_j$ in the original linear program. Consider two subsets $S, T \subseteq [n]$ such that $|S| + |T| \leq R - 1$. Then the following constraint is valid for $\calX$; i.e, all $x \in \calX$, the constraint is satisfied:
\begin{align*}
	\textstyle \prod_{i \in S}x_i\prod_{i \in T}(1-x_i)\left(\sum_{i = 1}^n a_{j, i}x_i - b_j\right) \leq 0.
\end{align*}

To {\em linearize} the above constraint, we expand the left side of the above inequality and replace each monomial with the corresponding $x_{S'}$ variable.  Then, we obtain the following :
\begin{align}
	\textstyle \sum_{T' \subseteq T}(-1)^{|T'|} \left(\sum_{i =1}^na_{j, i}x_{S \cup T' \cup \{i\}} - b_jx_{S \cup {T'}}\right) \leq 0. \label{inequ:SA}
\end{align}

%TODO: Some trivial facts

The $R$-th round of Sherali-Adams lift contains the above constraint for all $j, S, T$ such that $|S| + |T| \leq R-1$, and the trivial constraint that $x_{\emptyset} = 1$. For a polytope $\calP$ and an integer $R \geq 1$, we use $\SA(\calP, R)$ to denote the poltyope obtained by the $R$-th round Sherali-Adams lift of $\calP$. %We also view $\calP$ (resp.\ $\SA(\calP, R)$) as the polytope of feasible solutions to the linear program $\calP$ (resp.\ $\SA(\calP, R)$). 
For every $i \in [n]$, we identify the variable $x_{i}$ in the original LP and $x_{\{i\}}$ in a lifted LP. 

%\paragraph{Conditioning} 
Let $x \in \SA(\calP, R)$ for some linear program $\calP$ on $n$ variables and $R \geq 2$. Let $i \in [n]$ be an event such that $x_{i} > 0$; then we can define a solution $x' \in \SA(\calP, R-1)$ obtained from $x$ by ``conditioning" on the event $i$. For every $S \in {[n] \choose R-1}$, $x'_S$ is defined as
$x'_S:= \frac{x_{S \cup \{i\}}}{x_i}$.  We shall show that $x'$ will be in $\SA(\calP, R-1)$ (Property~\ref{property:conditioning-still-inside}).

It is useful to consider the ideal case where $x$ corresponds to a convex combination of integral solutions in $\calX$. Then we can view $x$ as a distribution over $\calX$. Conditioning on the event $i$ over the solution $x$ corresponds to conditioning on $i$ over the distribution $x$.  With this view, it is not hard to image the statements in the following claim (which we prove in the appendix) should hold:
\begin{restatable}{claim}{claimSA}
	\label{claim:SA}
	For some $x \in \SA(\calP, R)$ with $R \geq 2$, the following statements hold:
	\begin{enumerate}[label=(\ref{claim:SA}\alph*),leftmargin=*]
		\item $x_S \geq x_{S'}$ for every $S \subseteq S' \in {[n] \choose \leq R}$. \label{property:SA-subset}
		\item If $x_i  = 1$ for some $i \in [n]$, then $x_{\{i, i'\}} = x_{i'}$ for every $i' \in [n]$. \label{property:SA-if-event-happens}
		\item If every $\hat x \in \calP$ has $\hat x_i \leq \hat x_{i'}$, then $x_{\{i, i'\}} = x_i$. \label{property:SA-implication}
%		\item \sh{If every $\hat x \in \calP$ satisfies ${\sum}_{i' = 1}^n a_{i'}\hat x_{i'} \leq b$, then for every $i \in [n]$, we have ${\sum}_{i' = 1}^n a_{i'}x_{\{i, i'\}} \leq bx_i$.}\label{property:SA-multiplying}
	\end{enumerate}
	Letting $x'$ be obtained from $x$ by conditioning on some event $i \in [n]$, the following holds:
	\begin{enumerate}[label=(\ref{claim:SA}\alph*),leftmargin=*, start=4]
		\item $x'_i = 1$. \label{property:SA-conditioning}
		\item $x' \in \SA(\calP, R-1)$. \label{property:conditioning-still-inside}
		\item If $x_{i'} \in \{0, 1\}$ for some $i' \in [n]$, then $x'_{i'} = x_{i'}$. \label{property:SA-conditioning-keeps-01}
	\end{enumerate}
\end{restatable}
\sh{Keep in mind that the three properties \ref{property:SA-subset}, \ref{property:SA-conditioning}  and \ref{property:SA-conditioning-keeps-01} will be used over and over again, often without referring to them.  \ref{property:SA-conditioning} says that conditioning on $i$ will fix $x_i$ to 1. \ref{property:SA-conditioning-keeps-01} says that once a variable is fixed to $0$ or $1$, then it can not be changed by conditioning operations. 
}

\section{Reducing Directed Steiner Tree to Label-Consistent Subtree}
\label{sec:reduction}

In this section, we present a reduction from DST to LCST.  In \Cref{subsec:decomposition-trees}, we define a \emph{decomposition tree}, which corresponds to a recursive partitioning of a Steiner tree $T$ of $G$. We show that the DST problem is equivalent to finding a small cost decomposition tree. Due to the balanced-partition lemma (\Cref{cor:balanced-partition}), we can guarantee that decomposition trees have depth $O(\log k)$, a crucial property needed to obtain a quasi-polynomial-time algorithm.  Then in \Cref{subsec:dt-to-lcst} we show that the task of finding a small cost decomposition tree can be reduced to an LCST instance on a tree of depth $O(\log k)$.  Roughly speaking, for a decomposition tree to be valid, we require that the separator vertex appears in both parts of a partition: as a root in one part and possibly a non-root in the other.  This can be captured by the label-consistency requirement. 

We shall use $T$ to denote a Steiner tree in the original graph $G$, and $u, v$ to denote vertices in $G$. We use $\tau$ to denote a decomposition tree, and $\alpha, \beta$ to denote \emph{nodes} of a decomposition tree. $\bfT^0$ will be used for the input tree of the LCST instance. We use $\bfT$ for a sub-tree of $\bfT^0$ and $p, q, o$ for \emph{nodes} in $\bfT^0$.  The convention extends to variants of these notations as well. 

\subsection{Decomposition Trees}
\label{subsec:decomposition-trees}

We now define decomposition trees.  Recall that in the DST problem, we are given a graph $G = (V, E)$, a root $\fab{r} \in V$, and a set $K \subseteq V \setminus \{r\}$ of $k$ terminals.
\begin{definition}
	\label{def:decomposition}
	A decomposition tree $\tau$ is a rooted tree where each node $\alpha$ is associated with a vertex $\mu_\alpha \in V(G)$ and each leaf-node $\alpha$ is associated with an edge $e_\alpha \in E(G)$. Moreover, the following conditions are satisfied:
	\begin{enumerate}[label=(\ref{def:decomposition}\alph*), leftmargin=*]
		\item \label{property:root-of-decomposition-tree} $\mu_{\root(\tau)} = r$. 
		\item \label{property:mu-is-head-of-e} For every leaf $\beta$ of $\tau$, we have $\mu_\beta = \head(e_\beta)$.
		\item \label{property:contain-root} For every non-leaf $\alpha$ of $\tau$ and every child $\alpha_2$ of $\alpha$ with $\mu_{\alpha_2} \neq \mu_\alpha$ the following holds. There is a child $\alpha_1$ of $\alpha$ with $\mu_{\alpha_1} = \mu_\alpha$ such that $\mu_{\alpha_2} = \tail(e_{\beta})$ for some leaf $\beta \in V(\tau[\alpha_1])$.  In particular, this implies that $\alpha$ has at least one child $\alpha_1$ with $\mu_{\alpha_1} = \mu_\alpha$.
	\end{enumerate}
	The cost of a decomposition tree $\tau$ is defined as $\cost(\tau):={\sum}_{\alpha\text{ a leaf of }\tau}c(e_\alpha)$.  
\end{definition}

We say a vertex $v$ is \emph{involved} in a sub-tree $\tau[\alpha]$ of a decomposition tree $\tau$ if either $v = \mu_\alpha$ or there is a leaf $\beta$ of $\tau[\alpha]$ such that $v = \tail(e_{\beta})$.  So the second sentence in Property~\ref{property:contain-root} can be changed to the following: There is a child $\alpha_1$ of $\alpha$ with $\mu_{\alpha_1} = \mu_\alpha$ such that $\mu_{\alpha_2}$ is involved in $\tau[\alpha_1]$.

We show that the DST problem can be reduced to the problem of finding a small-cost decomposition tree of depth $O(\log k)$.  This is done in two directions.  \vspace*{-10pt}

\paragraph{From Directed Steiner Tree to Decomposition Tree.} We first show that the optimum directed Steiner tree $T^*$ of $G$ connecting $r$ to all terminals in $K$ gives a good decomposition tree $\tau^*$ of cost at most that of $T^*$, which we denote by $\opt$. Since we assumed costs of edges in $G$ satisfy triangle inequalities, we can assume every vertex $v \in V(T^*) \setminus (\{r\} \cup K)$ has at least two children in $T^*$.  This implies $|V(T^*)|\leq 2k$.  The decomposition tree $\tau^*$ can be constructed by applying \Cref{cor:balanced-partition} on $T^*$ recursively until we obtain trees with singular edges. Formally,  we set $\tau^* \gets \codt(T^*)$, where $\codt$ is defined in Algorithm~\ref{alg:codt}. Notice that the algorithm is only for analysis purpose and is not a part of our algorithm for DST.

\begin{algorithm}[h]
	\caption{$\codt(T)$} \label{alg:codt}
	\begin{algorithmic}[1]
		\If{$T$ consists of a single edge $(u, v)$}
			 \Return a node $\beta$ with $\mu_\beta = u$ and $e_\beta = (u, v)$ \label{State:create-beta}
		\Else
			\State create a node $\alpha$ with $\mu_\alpha = \root(T)$  \label{State:create-alpha} 
			\State apply \Cref{cor:balanced-partition} to find two rooted trees $T_1$ and $T_2$ with $\root(T_1) = \root(T)$ \label{State:balance-partition}
			\State $\tau_1 \gets \codt(T_1), \tau_2 \gets \codt(T_2)$ \label{State:codt-recurse}
			\State \Return the tree rooted at $\alpha$ with two sub-trees $\tau_1$ and $\tau_2$ \label{State:return-decomp-tree}
		\EndIf 
	\end{algorithmic}
\end{algorithm}

\begin{figure}
	\centering
	\includegraphics[width=0.7\textwidth]{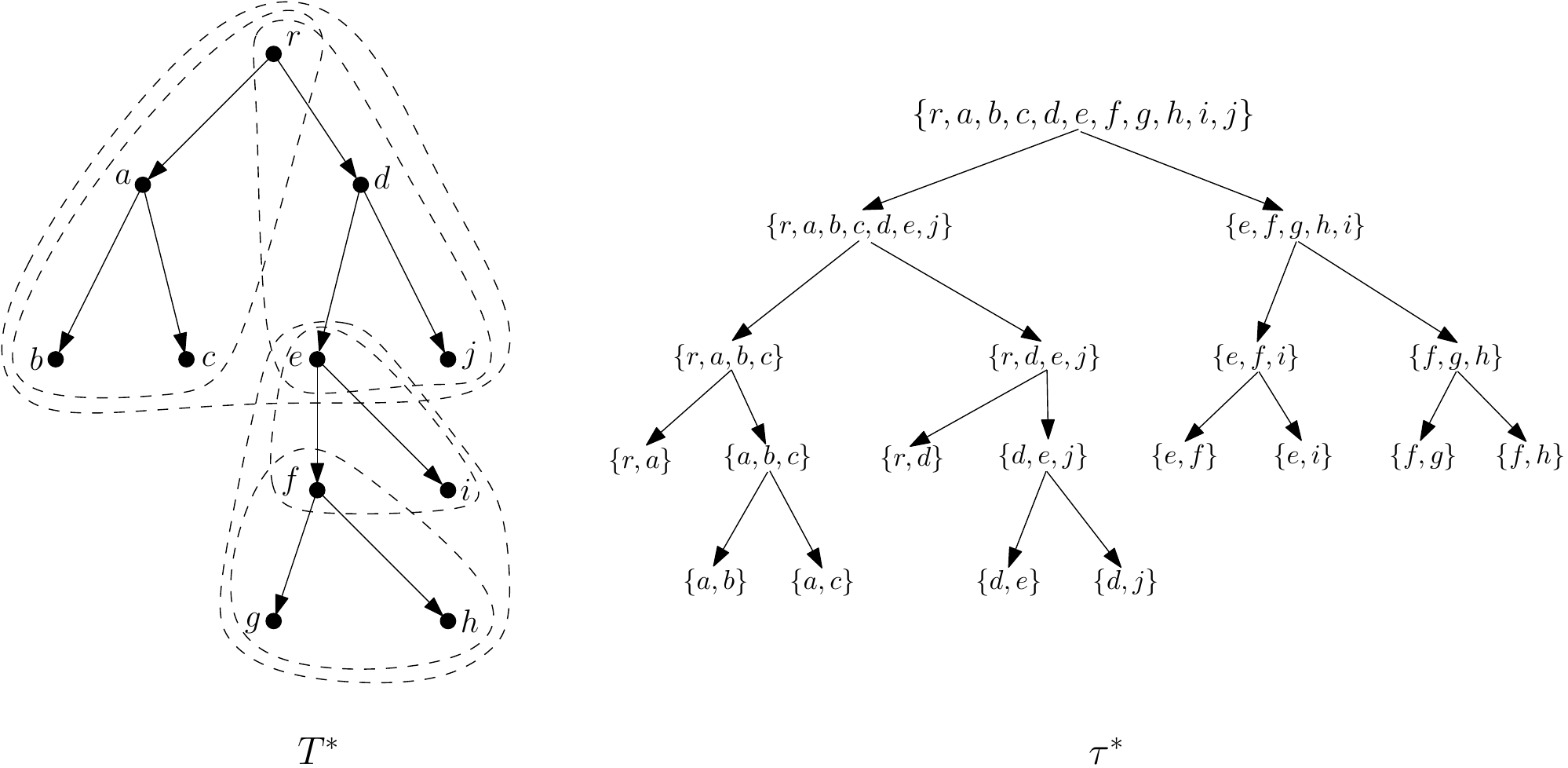}
	\caption{An example for construction of $\tau^*$. %The tree on the left-side is the optimum directed Steiner tree $T^*$, and the tree on the right side shows the correspondent decomposition tree $\tau^*$. 
	For each node $\tau^*$, the set denotes the vertices in the sub-tree of $T^*$ correspondent to the node; the $\mu$ value of the node is the first element in the set. For a leaf node, its $e$ value is the edge from the first element to the second element in the set.}
\end{figure}

\begin{restatable}{claim}{claimbinarydecompositiontree}
	\label{claim:binary-decomposition-tree}
	$\tau^*$  is a full binary decomposition tree of height $O(\log k)$ and cost $\opt$ that involves all vertices in $K$.   Moreover, for every $v \in K$, there is exactly one leaf $\beta$ of $\tau^*$ with $\tail(e_\beta) = v$.
\end{restatable} \vspace*{-10pt}

\paragraph{From Decomposition Tree to Directed Steiner Tree.} Now we show the other direction of the reduction.  The lemma we shall prove is the following:
\begin{restatable}{lemma}{lemmadecompositiontoSteiner}
	\label{lemma:decomposition-to-Steiner}
	Given a decomposition tree $\tau$ that involves all terminals in $K$, we can efficiently construct a directed Steiner tree $T$ in $G$ connecting $r$ to all terminals in $K$ with cost at most $\cost(\tau)$.
\end{restatable}

Thus, our goal is to find a decomposition tree of small cost involving all terminals in $K$. To do so, we construct an instance of the LCST problem. 

\subsection{Construction of LCST Instance}
	\label{subsec:dt-to-lcst}

	Let ${\bar h}$ be the $O(\log k)$ term in Claim~\ref{claim:binary-decomposition-tree} that upper bounds the height of $\tau^*$.   In the reduction, we shall ``collapse'' every $g := \ceil{\log_{\fab{2}}\log_{\fab{2}} k}$\fabr{Added base of log since we are outside O()} levels of a decomposition tree into one level; \sh{this is used to obtain the improvement of $\Theta(\log\log k)$ in the approximation ratio}. It motivates the definition of a \emph{twig}, which corresponds to a full binary tree of depth at most $g$ that can appear as a part of a decomposition tree:
	\begin{definition}
		A \emph{twig} is a rooted full binary tree $\eta$ of depth at most $g$, where  
		\begin{itemize}
			\item each $\alpha \in V(\eta)$ is associated with a $\mu_{\alpha} \in V(G)$, such that for every internal node $\alpha$ in $\eta$, at least one child $\alpha'$ of $\alpha$ has $\mu_{\alpha'} = \mu_\alpha$, and
			\item each leaf $\beta$ of $\eta$ \emph{may or may not} be associated with a value $e_{\beta} \in E(G)$; if $e_{\beta}$ is defined then $\head(e_\beta) = \mu_\beta$.
		\end{itemize}
	\end{definition}
	
	With the twigs defined, our LCST instance $\bfT^0$ is constructed by calling $\bfT^0 \gets \clt(r, 0)$, where $\clt$ is defined in Algorithm~\ref{alg:clt}. See Figure~\ref{fig:label-tree} for illustration of one recursion of $\clt$. 
	%The global labels are $K$, the set of terminals. 
	
	\begin{algorithm}
		\caption{$\clt(u, j)$}
		\label{alg:clt}
		\begin{algorithmic}[1]
			\State\label{State:create-root-p}create a new node $p$ with $c_p = 0, u_p = u$ and $\dem( p)  = \{\ell\}$ for a newly created local label $\ell$
%			\For{each vertex $v\in V(G)$ such that $(u, v) \in E(G)$ } \label{State:loop-leaf}
%				\State \label create a new node $o$ with $c_o = c(u, v)$,  $\ser(o)  =\{\ell, v\}$, $u_o = u$ and $f_o = (u, v)$
%				\State let $o$ be a child of $p$
%			\EndFor
			\If{$j < \ceil{{\bar h}/g}$}
				\For{each possible non-singular twig $\eta$ with $\mu_{\root(\eta)} = u$} \label{State:loop-recurse}
					\State \label{State:create-internal-q} create a node $q$ with $c_q = \sum_{\text{leaf } \beta \text{ of }\eta: e_\beta\text { defined }}c(e_\beta),  \eta_q = \eta, \ser(q) = \{\ell\}$, and $\dem(q) = \emptyset$ 
					\State let $q$ be a child of $p$
					\For{every leaf $\beta$ of $\eta$} \label{State:clt-loop-labels-for-leaves}
						\If{$e_\beta$ is defined}
							\State \textbf{if} $\tail(e_\beta) \in K$ \textbf{then} add \sh{the global label} $\tail(e_\beta)$ to $\ser(q)$ \label{State:add-global-label}
%							\State \label{State:create-leaf-o} create a node $o$ with $e_o = e_\beta, c_o = c(e_\beta), \ser(o) = \{\tail(e_\beta), \ell'\}$ and $\dem(o) = \emptyset$
%							\State $\bfT^q_\beta \gets$ the tree containing the single node $o$
						\Else
							\State $\bfT^q_\beta \gets \clt(\mu_\beta, j+1)$, let $\root(\bfT^q_\beta)$ be a child of $q$ \label{State:clt-recurse}
							\State create a new label $\ell'$,  add $\ell'$ to $\dem(q)$ and $\ser(\root(\bfT^q_\beta))$. \label{State:create-ellp-1}
						\EndIf
					\EndFor
					\For{every internal node $\alpha$ of $\eta$} \label{State:clt-loop-labels-for-consistency}
						\State let $\alpha_1$ be a child of $\alpha$ with $\mu_{\alpha_1} = \mu_\alpha$ and ${\alpha_2}$ be the other child
						\If{$\mu_{{\alpha_2}} \neq \mu_\alpha$ and $\nexists$ leaf $\beta$ of $\eta[\alpha_1]$ with $e_\beta$ defined and $\tail(e_\beta) = \mu_{\alpha_2}$}
							\State create a new label $\ell'$ and add it to $\dem(q)$ \label{State:create-ellp-2}
							\For {every leaf $\beta$ of $\eta[\alpha_1]$ with $e_\beta$ undefined, and $q'$ in $\bfT^q_{\beta}$} 
								\State \textbf{if} $\eta_{q'}$ has a leaf $\beta'$ with $e_{\beta'}$ defined and $\tail(e_{\beta'}) = \mu_{\alpha_2}$ \textbf{then} add $\ell'$ to $\ser(q')$
							\EndFor
						\EndIf
					\EndFor
				\EndFor
			\EndIf
			\State \Return the tree rooted at $ p$
		\end{algorithmic}
	\end{algorithm}
%\fabr{In Alg. 2 line 3: shall $f_o$ be $e_o$?}	
	
	\begin{remark}
	 	The $u$ and $\eta$ values of nodes in $\bfT^0$ are irrelevant for the LCST instance. They will, however, help us in mapping the decomposition tree to its corresponding solution to LCST.
	\end{remark}
	
	Notice that there are two types of nodes in $\bfT^0$: (1) $p$-nodes are those created in Step~\ref{State:create-root-p} and (2) $q$-nodes are those created in Step~\ref{State:create-internal-q}. We always use $p$ ($q$, resp.) and its variants to denote $p$-nodes ($q$-nodes resp.). 

\fabr{Where are the global labels in this discussion?}
We give some intuition behind the construction of $\bfT^0$.  We can partition the edges of a decomposition tree $\tau$ into an $O(\bar h/g)$-depth tree $\bfH$ of twigs. For each $\eta$ in the tree, we apply the following operation. First, we replace $\eta$ with a node $q$ with $\eta_q = \eta$. Second, we insert a virtual parent $p$ of $q$ with $u_p = \mu_{\root(\eta)}$ between this $q$ and its actual parent.  Then it is fairly straightforward to see that we can find a copy of this resulting tree in $\bfT^0$. Thus, we reduced the problem of finding $\bfH$ (and thus $\tau$) to the problem of finding a subtree $\bfT$ of $\bfT^0$.  The label-consistency requirements shall guarantee that $\bfT$ will correspond to a valid $\tau$. In particular, the demand label $\ell$ for a node $p$ created in Step~\ref{State:create-root-p} guarantees that if $p$ is selected then we shall select at least one child of $p$. The demand labels created in Step~\ref{State:create-ellp-1} for a node $q$ guarantee that if $q$ is selected, then all its children must be selected, while the demand labels created in Step~\ref{State:create-ellp-2} guarantee Property~\ref{property:contain-root} of $\tau$.   \sh{The set of global labels is exactly $K$. In Step~\ref{State:add-global-label}, we add a global label $v \in K$ to $q$ if $\eta_q$ contains a leaf $\beta$ with $\tail(e_\beta) = v$.}
		
	\begin{figure}
		\centering
		\includegraphics[width=0.5\textwidth]{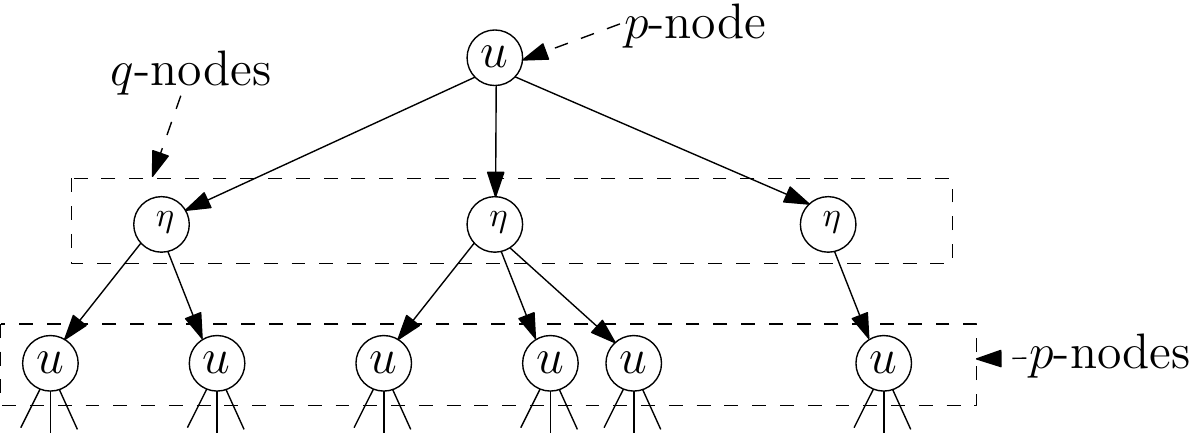}
		\caption{Nodes created in one recursion of $\clt$. Each $p$-node has a $u_p$ value, and each $q$-node is associated with a twig $\eta_q$ with $\mu_{\root(\eta_q)}$ being the $u$ value of its parent $p$-node. Each child $p'$ of $q$ corresponds to a leaf $\beta$ of $\eta_q$ with $e_\beta$ undefined.}
		\label{fig:label-tree}
	\end{figure}

	A simple observation we can make is the following:
	\begin{restatable}{claim}{HugeTree}
		\label{claim:huge-tree}
		$\bfT^0$ is a rooted tree with $n^{O(\log^2 k/\log \log k)}$ vertices and height $O(\bar h/g) = O(\log k/\log\log k)$, where $n = |V(G)|$. 
	\end{restatable}
	Also, it is easy to see that a node $p$ will have exactly one demand label, while a node $q$ can have up to $O(2^g)$ demand labels. So,  we have $s:=\max_{p \in V(\bfT^0)} |\dem(v)| = O(2^g) = O(\log k)$.
	
	We then show that the problem of finding a decomposition tree can be reduced to that of finding a label-consistent subtree of $\bfT^0$.  Again, this is done in two directions.  
	
\paragraph{From Decomposition Tree to Label-Consistent Subtree}
	To show that there is a good label-consistent subtree $\bfT^*$ of $\bfT^0$, we need to construct a tree of twigs from $\tau^*$.  This is done as follows. For every $i = 0, 1, 2, \cdots$, and every internal node $\alpha$ in $\tau^*$ of depth $ig$, we create a twig rooted at $\alpha$ containing all descendants of $\alpha$ at depth $ig, ig+1, ig+2, \cdots, (i+1)g$. Let $\calV$ be the set of twigs created. A rooted tree $\bfH$ over $\calV$ can be naturally defined: a twig $\eta$ is a parent of $\eta'$ if and only if $\root(\eta')$ is a leaf in $\eta$. So, $\bfH$ has depth at most $\ceil{{\bar h}/g}$.
	
	\begin{algorithm}
		\caption{$\colcst(p, \eta)$}
		\label{alg:colcst}
		\begin{algorithmic}[1]
			\State add $p$ and the child $q$ of $p$ with $\eta_q = \eta$ to $\bfT^*$ \Comment{such a $q$ exists since $\mu_{\root(\eta)} = u_p$}
			\For{every leaf $\beta$ of $\eta$ such that $e_\beta$ is not defined}
				\State let $\eta'$ be the twig in $\calV$ with $\root(\eta') = \beta$
				\State $\colcst(\root(\bfT^q_\beta), \eta')$ \label{State:colcst-recurse}
			\EndFor
		\end{algorithmic}
	\end{algorithm}
	
	$\bfT^*$ can be found naturally by calling $\colcst(\root(\bfT^0), \root(\bfH))$ (with $\bfT^*$ being empty initially), where $\colcst$ is defined in Algorithm~\ref{alg:colcst}, and the trees $\bfT^q_\beta$ are as defined in Algorithm~\ref{alg:clt}. 
	The recursive procedure takes two parameters: a node $p$ in $\bfT^0$ and a twig $\eta \in \calV$. It is guaranteed that $u_p = \mu_{\root(\eta)}$: The root recursion satisfy this condition since $u_{\root(\bfT^0)} = \mu_{\root(\root(\bfH))} = r$; in Step~\ref{State:colcst-recurse}, we also have $u_{\root(\bfT^q_\beta)} = \mu_\beta = \mu_{\root(\eta')}$.  The tree can be constructed as $\bfH$ has depth at most $\ceil{\bar h/g}$. 
	Again, this algorithm is only for analysis purpose and is not a part of our algorithm for DST. We prove in the appendix the following lemma.
	
	\begin{restatable}{lemma}{lemmadecptoLCST}
		\label{lemma:decp-to-LCST}
		$\bfT^*$ is a label-consistent sub-tree of $\bfT^0$  with cost exactly $\cost(\tau^*) = \opt$. Moreover, all global labels in $K$ are supplied by $\bfT^*$.
	\end{restatable}

\paragraph{From Label-Consistent Subtree to Decomposition Tree.} The following lemma gives the other direction, and its proof will be deferred to the appendix. 
	\begin{restatable}{lemma}{lemmaLCSTtodecompositiontree}
		\label{lemma:LCST-to-deomposition-tree}
		Given any feasible solution $\bfT$ to the LCST instance $\bfT^0$, in time $\poly(|V(\bfT)|)$ we can construct a decomposition tree $\tau$ with $\cost(\tau) = \cost(\bfT)$. Moreover, if a global label $v \in K$ is supplied by $\bfT$, then $\tau$ involves $v$. 
	\end{restatable}

\paragraph{Wrapping up.} We prove the following theorem in the next section.  Recall that $N$ and $h$ are respectively the size and height of the input tree $T^0$ to the LCST instance, and $k$ is the number of global labels. 
	\begin{restatable}{theorem}{approxlabelconsistent}
		\label{thm:approx-label-consistent}
		There is an $(shN)^{O(sh^2)}$-time $O(h \log k)$-approximation algorithm for the Label-Consistent Subtree problem where $s:=\max_{v \in V(T^0)}|\dem(v)|$.
	\end{restatable}
	
	With this theorem at hand, we can now finish our $O(\log^2k/\log\log k)$-approximation for DST that runs in quasi-polynomial time.  Given a DST instance, we shall construct the LCST instance $\bfT^0$ of size $N = n^{O(\log^2 k/\log\log k)}$ and height $h = O(\log k/\log\log k)$ as in Algorithm~\ref{alg:clt}.  Notice that for the LCST instance, we have $s := \max_{p \in V(\bfT^0)}|\dem(p)| = O(2^g) = O(\log k)$. By \Cref{claim:binary-decomposition-tree} and \Cref{lemma:decp-to-LCST}, there is a solution $\bfT^*$ to the LCST instance $\bfT^0$ of cost at most $\opt$.  Applying \Cref{thm:approx-label-consistent}, we can obtain a feasible solution $\bfT$ of cost at most $O(h \log k)\cdot \opt = O(\log ^2k/\log\log k) \cdot \opt$ in time $(shN)^{O(sh^2)} = n^{O(\log^5 k)}$ (as $s = O(\log k)$). Applying \Cref{lemma:LCST-to-deomposition-tree} and \Cref{lemma:decomposition-to-Steiner}, we can obtain a Directed Steiner tree $T$ in $G$ of cost at most  $O(\log ^2k/\log \log k) \cdot \opt$ connecting $r$ to all terminals in $K$.  This gives a $O(\log^2 k/\log \log k)$-approximation for DST in running time $n^{O(\log^5k)}$, finishing the proof of \Cref{thm:main}.

\section{Approximation Algorithm for Label-Consistent Subtree}
\label{sec:algo}

%TODO: need an overview. Rothvoss's result. 

%TODO: s in the base in running time. Change when I have token.

%	There is an $(shN)^{O(sh^2)}$-time $O(h \log k)$-approximation algorithm for the Label-Consistent Subtree problem where $s:=\max_{v \in V(T^0)}|\dem(v)|$.

The goal of this section is to prove Theorem~\ref{thm:approx-label-consistent}, which is repeated below.  Since we are not dealing with the original DST problem any more, we use $T^0, T$ for trees and $u, v$ for nodes in this section.
\approxlabelconsistent*

\subsection{Redefining the LCST Problem}
We shall first simplify the input instance w.l.o.g in the following ways that will make our presentation much cleaner.  Indeed, some properties are already satisfied by the LCST instance reduced from the DST problem; however we want to make Theorem~\ref{thm:approx-label-consistent} as general as possible and thus we do not make these assumptions in the theorem statement.
\begin{enumerate}[leftmargin=*]
	\item We can assume for every two distinct nodes $u$ and $v$, $\dem(u)$ and $\dem(v)$ are disjoint. If some local label $\ell$ appears in $\dem(u)$ for $t \geq 2$ different nodes $u$, we can make $t$ copies of $\ell$ and let each copy be contained in $\dem(u)$ for exactly one $u$. We can replace the appearance of $\ell$ in some $\ser(v)$ with the $t$ copies.
	\item We can assume the demand labels are only at the internal nodes. Suppose a leaf $v$ has $\ell \in \dem(v)$. If $\ell \in \ser(v)$, then $\ell$ can be removed from $\dem(v)$; otherwise $v$ can never be selected thus can be removed from $T^0$. 
	\item We can assume that the service labels are only at the leaves and each leaf contains exactly one service label. A leaf without a service label can be removed.  For a non-leaf $v$ with $\ser(v) \neq \emptyset$, we can attach $|\ser(v)|$ leaves of cost $0$ to $v$ and distribute the service labels to the newly added leaves.  Similarly, if a leaf $v$ has $|\ser(v)| > 1$, we can attach $|\ser(v)|$ new leaves to $v$.
%	\item We can assume that each internal node $v \in V(T^0)$ has $|\dem(v)| = 1$: If $|\dem(v)| > 1$ we can replace $v$ with a path of $|\dem(v)|$ nodes, each containing one demand label.  This will increase the height and the size of the $T^0$ by $O(1)$ since we assumed $\max_{v \in V(T^0)}|\dem(v)| \leq O(1)$. Also, if $\dem(v) = \emptyset$ for some $v$, we can create a new local label $\ell$, add it to $\dem(v)$, and attach a leaf $u$ to $v$ with $\ser(u) = \{\ell\}$. 
\end{enumerate}
Notice that the above operations do not change the set $K$ of global labels and $s = \max_{v \in V(T^0)}|\dem(v)|$. \vspace*{-10pt}

With the above operations and simplifications, we can redefine the LCST instance.   Let $V^\leaf$ and $V^\internal$ respectively be the sets of leaves and internal nodes of $T^0$.  For every node $v \in V^\internal$, let $\Lambda_v$ be the set of children of $v$. For every $v \in V(T^0)$,  let $\Lambda^\leaf_v = V(T^0[v]) \cap V^\leaf$ be the set of descendants of $v$ that are leaves. 

%For every $v \in V^\internal$,we let $a_v$ be the unique label in $\dem(v)$.
For every $v \in V^\leaf$, let $a_v$ be the unique label in $\ser(v)$.  From now on we shall not use the notation $\ser(\cdot)$ anymore.
%
%We are given a rooted tree ${T^0}$ on a set  ${V(T^0)}$ of $N = |V(T^0)|$ nodes; let ${\root(T^0)} \in {V(T^0)}$ be the root of $T^0$. Let $h$ be the height of ${T^0}$, i.e, the length of the longest root-to-leaf path, and 
%Moreover,  each node $u \in {V(T^0)}$ is associated with a cost $c_u \geq 0$. 
%We are also given a set $L$ of labels, among which there is a subset $K \subseteq L$ of size $k = |K|$ corresponding to the set of terminals in the original directed Steiner tree problem. Each node $u \in {V(T^0)}$ has a label $a_u \in L$ such that the labels of nodes in $V^\internal$ are all distinct.  A rooted subtree $T$ of $T^0$ is a connected sub-graph of $T^0$ that contains the root ${\root(T^0)}$. 
%
Thus, a rooted subtree $T$ of ${T^0}$ with $\root(T) = \root(T^0)$ is label-consistent if, for every $u \in V(T) \cap V^\internal$ and $\ell \in \dem(u)$, there is a node $v \in V(T) \cap \Lambda^\leaf_u$ with $a_v = \ell$. 

The goal of the problem is to find the minimum cost label-consistent subtree $T$ of $T^0$ that provides all the global labels, i.e, that satisfies for all $\ell \in K$ there exists a $v \in V(T) \cap V^\leaf$ with $a_v = \ell$. % (due to the label consistence property, there exists a $v \in V(T)\cap V^\leaf$ with $a_v = \ell$).
Recall that we are given a node-cost vector $c \in \R_{\geq 0}^{V(T^0)}$. The cost of a sub-tree $T$ of $T^0$, denoted as $\cost(T)$, is defined as $\cost(T):={\sum}_{v \in V(T)}c_v$.

We consider the change in the size and height of $T^0$ after we applied the above operations.  Abusing notations slightly,  we shall use $N'$ and $h'$ to store the size and height of the old $T^0$ (i.e, the $T^0$ before we apply the operations), and $N$ and $h$ be the size and height of the new $T^0$ (i.e, the $T^0$ after we apply the operations).  Notice that we only added leaves to $T^0$.  Thus, we have $h \leq h' + 1$. The number of internal nodes in the new $T^0$ is at most $N'$.  A leaf $v$ is relevant only when it is providing a label that are in $\dem(u)$ for some ancestor $u$ of $v$. If a node has many leaf children with the same service label, we only need to keep the one with the smallest cost. Since each $u$ has $|\dem(u)| \leq s$ and the height of the old $T^0$ is $h'$, we can assume that the number of leaves in the  new $T^0$ is at most $s(h'+1)N'$.  So $N \leq s(h'+1)N' + N' = O(sh'N')$.

Let $T^*$ be the optimum tree for the given instance. Let $\opt$ be the cost of the $T^*$, i.e, $\opt = \cost(T^*)$.\footnote{We remark that it is easy to check whether a valid solution exists or not: an $u \in V^\internal$ is useless if for some $\ell \in \dem(u)$ there is no $v \in \Lambda^\leaf_u$ with $a_v = \ell$. We repeatedly remove useless nodes and their descendents until no such nodes exist.  There is a valid solution iff the remaining $T^0$ provides all labels in $K$. So we can assume the instance has a valid solution.} As every local label appears only once in $V^\internal$, we can assume that for every $\ell \in L$, there is at most one node $v \in V(T^*) \cap V^\leaf$ with $a_v = \ell$: if there are multiple such nodes $v$, we can keep one without violating the label-consistency condition and that all global labels are provided. Thus additionally we can assume $T^*$ satisfies the following conditions:
\refstepcounter{theorem}
\label{properties:T^*}
\begin{enumerate}[label=(\ref{properties:T^*}\alph*), leftmargin=*]
	\item For every $\ell \in K$, there is exactly one node $v \in V(T^*) \cap V^\leaf$ such that $a_v = \ell$. \label{property:one-for-K}
	\item For every $\ell \in L \setminus K$, there is at most one node $v \in V(T^*) \cap V^\leaf$ such that $a_v = \ell$. \label{property:at-most-one}
\end{enumerate}

%TODO define \cost(T') = 

The main theorem we shall prove is the following 	
\begin{theorem}
	\label{thm:find-tree}
	There is an $(sN)^{O(sh^2)}$-time algorithm that outputs a random label-consistent tree $\tilde T$ such that, 
	$\E\left[c(\tilde T)\right] \leq \opt$,
	and for every $\ell \in K$, we have
	$\Pr\left[\exists v \in V^\leaf \cap V(\tilde T): a_v = \ell\right] \geq \frac{1}{h+1}$.
\end{theorem}

With theorem~\ref{thm:find-tree}, we can finish the proof of Theorem~\ref{thm:approx-label-consistent}.
\begin{proof}[Proof of Theorem~\ref{thm:approx-label-consistent}]
	We run $O(h \log k)$ times the algorithm stated in Theorem~\ref{thm:find-tree} and let $T'$ be the union of all the trees $\tilde T$ produced. It is easy to see that $T'$ is always label-consistent. The expected cost of $T'$ is 
	\begin{align*}
		\E\left[\cost(T')\right] \leq O(h \log k)\opt.
	\end{align*}
	If the $O(h \log k)$ term is sufficiently large, by the union bound, we can obtain
	\begin{align}
	\label{inequ:prob-good}
	\Pr \left[\forall \ell \in K, \exists v \in V^\leaf \cap V(T'), a_v = \ell\right] \geq 1/2.
	\end{align}
	
	We repeatedly run the above procedure until $\forall \ell \in K, \exists v \in V^\leaf \cap V(T'), a_v = \ell$ happens and output the tree $T'$ satisfying the property. Let $T^{\textrm{final}}$ be this tree. Then we have $\E\left[\cost(T^{\textrm{final}})\right] \leq O(h \log k)\opt$ due to \eqref{inequ:prob-good}. In expectation we only need run the procedure twice.
	
	Thus, we obtain an $O(h \log k)$-approximation algorithm for LCST. The running time of the algorithm is $(sN)^{O(sh^2)} = (sh'N')^{O(sh'^2)}$.  Recall that $h'$ and $N'$ are the height and size of $T^0$ before we applied the operations; thus the theorem follows.
\end{proof}\vspace*{-10pt}

Thus, our goal is to prove Theorem~\ref{thm:find-tree}.
%\subsection{Overview of Algorithm}  
Our algorithm is very similar to that of \cite{Rothvoss11} for GST on trees.  We solve the lifted LP relaxation for the LCST problem and then round the fractional solution via a recursive procedure. In the procedure, we focus on some sub-tree $T^0[u]$, and we are given a set $L'$ of labels that must appear in $\tilde T[u]$, where $\tilde T$ is our output tree. We are also given a lifted LP solution $x$; we can restrict $x$ on the tree $T^0[u]$.  The set $L'$ of labels appear in $T^0[u]$ fully according to $x$. Then, for every $\ell \in L'$, we randomly choose child $v$ of $u$ that is responsible for this $\ell$ and then apply some conditioning operations on $x$. We recursively call the procedure for the children of $u$.  This way, we can guarantee that the tree $\tilde T$ we output is always label-consistent. Finally, we show that each global label $v \in K$ appears in $\tilde T$ with large probability, using the technique that is very similar to that of \cite{Rothvoss11}.

\subsection{Basic LP Relaxation} The remaining part of the section is dedicated to the proof of Theorem~\ref{thm:find-tree}.  We formulate an LP relaxation that aims at finding the $T^*$, where the variables of the LP are indexed by $\bbD = {V(T^0)} \cup ({V(T^0)} \times L)$. We view every element in $\bbD$ also as an event. Supposedly, an event $u \in {V(T^0)}$ happens if and only if $u \in V(T^*)$, and an event $(u, \ell) \in {V(T^0)} \times L$ happens if and only if $u \in V(T^*)$ and $\Lambda^\leaf_u \cap V(T^*)$ has a node with label $\ell$ (such a node is unique if it exists by Properties~\ref{property:one-for-K} and \ref{property:at-most-one}).  For every $e \in \bbD$, $x_e \in \{0, 1\}$ is supposed to indicate whether event $e$ happens or not. Then the following linear constraints are valid:

\noindent\begin{minipage}{0.46\textwidth}
	\begin{align}
	x_v &\leq x_u, &\quad &\forall u \in V^\internal, v \in \Lambda_u  \label{LPC:parent-child} \\[10pt]
	x_{(u, \ell)} &\leq x_u, &\quad &\forall u \in V(T^0), \ell \in L  \label{LPC:uell-less-than-u}\\[5pt]
	x_{(u, \ell)} &= x_u, &\quad &\forall u \in V^\internal, \ell \in \dem(u) \label{LPC:label-of-u}\\
	x_{(v, a_v)} &= x_v, &\quad &\forall v \in V^\leaf \label{LPC:label-of-v}
	\end{align}
\end{minipage}\hfill
\begin{minipage}{0.53\textwidth}
	\begin{align}
	x_{(u, \ell)} &= \sum_{v \in \Lambda_u}x_{(v, \ell)}, &\quad &\forall u \in V^\internal, \ell \in L \label{LPC:sum-over-children}\\
	x_{(v, \ell)} &= 0, &\quad &\forall v \in V^\leaf, \ell \neq a_v \label{LPC:leaf} \\[5pt]
	x_{({\root(T^0)}, \ell)} &= 1, &\quad &\forall \ell \in K \label{LPC:root}\\ 
	\nonumber
	\end{align}		
\end{minipage}\bigskip

\eqref{LPC:parent-child} holds since $T^*$ is rooted sub-tree of $T^0$ with $\root(T^*) = \root(T^0)$, \eqref{LPC:uell-less-than-u} holds by definition of events,  \eqref{LPC:label-of-u} follows from that $T^*$ is label-consistent, and \eqref{LPC:label-of-v} holds trivially. \eqref{LPC:sum-over-children} follows from Properties~\ref{property:one-for-K} and \ref{property:at-most-one}. \eqref{LPC:leaf} holds trivially and \eqref{LPC:root} follows from  Property~\ref{property:one-for-K}.  %Also, the cost of the solution $\sum_{u \in V(T^0)} c_u x_u$ is at most $\opt$.  

Let $\calP$ be the polytope containing all vectors $x \in [0, 1]^{\bbD}$ satisfying constraints~\eqref{LPC:parent-child} to \eqref{LPC:root}. The following simple observation can be made:
\begin{claim}
	\label{claim:sum-up-leaf-decedants}
	For every $x \in \calP$, $u' \in V(T^0)$, and $\ell' \in L$,  we have $\sum_{v \in \Lambda^\leaf_{u'}} x_{v, \ell'} = x_{u', \ell'}$.
\end{claim}
\begin{proof} 
	The claim holds trivially if $u' \in V^\leaf$. When $u' \notin V^\leaf$,  summing up \eqref{LPC:sum-over-children} over all internal nodes $u$ in $T^0[u']$ and $\ell = \ell'$ gives the equality. 
\end{proof}

\subsection{Rounding a Lifted Fractional Solution} Let $R = O(sh^2)$ be large enough. Since $\calP$ contains an integral solution of cost at most $\opt$, we can find a solution $x^* \in \SA(\calP, R)$ with ${\sum}_{v \in {V(T^0)}}c_v x^*_v \leq \opt$ in running time $|\bbD|^{O(sh^2)} = (sN)^{O(sh^2)}$. 
\begin{remark}
	Indeed, our algorithm only need\fab{s} to use variables that correspond to paths of $\bfT^0$ starting at the root. Using this one can remove a $\log k/\log\log k$ factor from the exponent of the running time.  However, we choose to use the Sherali-Adams hierarchy as it is much easier to describe.
\end{remark}

In the main rounding algorithm (Algorithm~\ref{alg:main}), we let $\tilde V = \emptyset$ initially and call $\solve({\root(T^0)}, \break \dem(\root(T^0)), x^*)$, as described in Algorithm~\ref{alg:solve}. We output the subtree $\tilde T$ of $T^0$ induced by $\tilde V$.

\noindent\begin{minipage}[t]{0.38\textwidth}
	\begin{algorithm}[H]
		\caption{Main Rounding}
		\label{alg:main}
		\ Given: $x^* \in \SA(\calP, R)$
		
		\ Output: a label-consistent tree $\tilde T$ \vspace*{-8pt}
		
		\rule{\textwidth}{0.5pt}
		\begin{algorithmic}[1]
			\State $\tilde V \gets \emptyset$
			\State $\solve({\root(T^0)}, \dem(\root(T^0)), x^*)$
			\State \Return the tree $\tilde T$ induced by $\tilde V$
		\end{algorithmic}
	\end{algorithm}
\end{minipage}\hfill
\begin{minipage}[t]{0.58\textwidth}
	\begin{algorithm}[H]
		\caption{$\solve(u, L', x)$}
		\label{alg:solve}
		\begin{algorithmic}[1]
			\State $\tilde V \gets \tilde V \cup \{u\}$ \label{State:add-u-to-tilde-V}
			\State \textbf{if} $u\in V^\leaf$ \textbf{then} \Return
			\State let $S_v \gets \emptyset$ for every $v \in \Lambda_u$
			\For{every $\ell \in L'$} \label{State:first-loop}
			\State randomly choose a child $v$ of $u$, so that $v$ is chosen with probability $x_{(v, \ell)}$ (see Property~\ref{property:sampling-well-defined}) \label{State:sampling}
			\State $S_v \gets S_v \cup \{\ell\}$
			\State $x \gets x$ conditioned on the event $(v, \ell)$ \label{State:conidition-on-vl}
			\EndFor
			\For{every $v \in \Lambda_u$, with probability $x_v$,} \label{State:second-loop}
			\State $\solve(v, S_v \cup \dem(v), x\text{ conditioned on event } v)$ \label{State:recurse}  %TODO: for leaves $v$, dem(v) = \emptyset.
			\EndFor
		\end{algorithmic}
	\end{algorithm}
\end{minipage}\bigskip

In the recursive algorithm $\solve(u, L', x)$, $u$ is the current node we are dealing with. %, which will be added to $\tilde T$.  
$L'$ is the set of labels that must be supplied in $\tilde T[u]$; in particular, we shall guarantee that $\dem(u) \subseteq L'$.  $x$ is the LP hierarchy solution that is passed to $u$, which satisfies $x_u = 1$ and $x_{(u, \ell)} = 1$ for every $\ell \in L'$ (Property~\ref{property:x-good} in Claim~\ref{claim:simple-things} that appears later). We add $u$ to $\tilde V$ in Step~\ref{State:add-u-to-tilde-V}; thus the final $\tilde T$ contains the set of nodes for which we called $\solve$.

If $u \in V^\leaf$, %\eqref{LPC:leaf} guarantees that $L'$ is either $\emptyset$ or $\dem(u)$ and thus 
we then do nothing; so focus on the case $u \notin V^\leaf$. To guarantee that a label $\ell \in L'$ is supplied in $\tilde T[u]$, we need to specify one child $v$ of $u$ such that $\tilde T[v]$ supplies $\ell$; we say that $v$ is responsible for this label $\ell$.  This is done via a random procedure by using the solution $x$ as a guide: the probability that $v$ is chosen is exactly $x_{(v, \ell)}$ (Step~\ref{State:sampling}). We shall show that ${\sum}_{v \in \Lambda_u}x_{v, \ell} = 1$ (Property~\ref{property:sampling-well-defined})  and thus the process is well-defined. After choosing the $v$ for this $\ell \in L'$, we update $x$ by conditioning on the event $(v, \ell)$ (Step~\ref{State:conidition-on-vl}). So far the number of nested conditioning operations we apply on $x$ is $|L'|$; we will see soon that $|L'|$ is small and thus we can apply these operations.

For every $v \in \Lambda_u$, let $S_v$ be the set of labels in $L'$ that $v$ is responsible for.  In Loop~\ref{State:second-loop}, we independently and recursively call $\solve$ on the children of $u$. Notice that $x_v$ is the extent to which $v$ is included in $V(T^*)$. So we only call $\solve$ on $v$ with probability $x_v$; the LP solution passed to the sub-recursion is $x$ conditioned on the event $v$. In particular if $S_v \neq \emptyset$ then $x_v = 1$.  We remark that the conditioning operations for all children $v$ of $u$ are done ``in parallel'' and thus we \fab{``lose only 1 level'' of our Sherali-Adams lifting}. %, even though $|\Lambda_u|$ might be large.
\medskip

We now analyze the algorithm. To prove Theorem~\ref{thm:find-tree}, we need to show that $\tilde T$ is a label-consistent subtree with small expected cost; moreover, every label $\ell \in K$ is provided by $\tilde T$ with large \fab{enough} probability. %We now analyze the recursive algorithm formally. 
Let us first assume that the number $R$ of rounds is large enough so that all the conditioning operations can be applied. We start from some simple observations for the algorithm.

\begin{claim}\label{claim:simple-things}
	For every recursion of $\solve$ that the algorithm invokes,
	\begin{enumerate}[label=(\ref{claim:simple-things}\alph*),leftmargin=*]
		\item \label{property:x-good} at the beginning the recursion, we have $x_u = 1$ and $x_{(u, \ell)} = 1$ for all $\ell \in L'$, and
		\item \label{property:sampling-well-defined} the random sampling process in Step~\ref{State:sampling} is well-defined: we have ${\sum}_{v \in \Lambda_u}x_{(v, \ell)} = 1$ before the step.
	\end{enumerate}
\end{claim}
\begin{proof}
	\ref{property:x-good} holds for the root recursion as
	%For simplicity, let Condition (*) be that $x_u = 1$ and $x_{(u, \ell)} = 1$ for all $\ell \in L'$. (*) holds at the beginning of the root recursion since
	\eqref{LPC:root} implies $x^*_{\root(T^0)} = 1$ and \eqref{LPC:label-of-u} implies $x^*_{{\root(T^0)}, \ell} = x^*_{\root(T^0)} = 1$ for every $\ell \in \dem(\root(T^0))$.
	
	Now assume \ref{property:x-good} holds for some recursion for $u \notin V^\leaf$. So, at the beginning of an iteration of Loop~\ref{State:first-loop}, we have $x_{u, \ell} = 1$ for every $\ell \in L'$. Thus, by \eqref{LPC:sum-over-children}, we have $\sum_{v \in \Lambda_u} x_{v, \ell} = 1$, implying \ref{property:sampling-well-defined} for this recursion. 
	%, and (*) holds at the beginning of each sub-recursion.
	
	Since we conditioned on the event $(v, \ell)$ in Step~\ref{State:conidition-on-vl} after adding $\ell$ to $S_v$, we have $x_{(v, \ell)} = 1$  for every $v \in \Lambda_u$ and $\ell \in S_v$  after finishing Loop~\ref{State:first-loop}. (Notice that Property \ref{property:SA-conditioning-keeps-01} says that once a variable has value 0 or 1, conditioning operations do not change its value.) Focus on Step~\ref{State:recurse} for some $v \in \Lambda_u$, and let $x'$ be the $x$ passed to the sub-recursion, i.e, $x'$ is obtained from $x$ by conditioning on the event $v$. Then we have that $x'_{v} = 1$ and $x'_{(v, \ell)} = 1$ for every $\ell \in S_v$. Also, $x'_{(v, \ell)} = x'_v = 1$ for every $\ell \in \dem(v)$ by \eqref{LPC:label-of-u}. Since $L' = S_v \cup \dem(v)$ in the sub-recursion of $\solve$ for $v$, \ref{property:x-good} holds for the sub-recursion for $v$.
\end{proof}

\begin{claim}
	\label{claim:label-consistent}
	The tree $\tilde T$ returned by Algorithm~\ref{alg:main} is label-consistent.
\end{claim}
\begin{proof}
	When we call $\solve$ for an $u$, it is guaranteed that $\dem(u) \subseteq L'$ (by Step 2 in Algorithm~\ref{alg:main} and Step~\ref{State:recurse} in Algorithm~\ref{alg:solve}). By the way we construct $S_v$'s in Loop~\ref{State:first-loop} of Algorithm~\ref{alg:solve}, each label $\ell \in \dem(u)$ will be passed down all the way to some leaf node $v \in \Lambda^\leaf_u$. By Property~\ref{property:x-good} for the recursion of $\solve$ for $v$, we must have $x_{v, \ell} = 1$ at the beginning of this recursion. Then by \eqref{LPC:leaf}, $\ell = a_v$ must hold.
\end{proof}

\begin{claim}
	\label{claim:rounds}
	If $R = O(sh^2)$ is large enough, then all the conditioning operations can be performed.
\end{claim}
\begin{proof}
	Notice that for the recursion of $\solve$ for $u$, the size of $|L'|$ passed to the recursion is at most $s(\mathrm{depth}(u) + 1)$, where $\mathrm{depth}(u)$ is the depth of $u$ in the tree $T^0$, i.e, the distance from ${\root(T^0)}$ to $u$. This holds since in a recursion of $\solve$ for $u$, $S_v$'s are subsets of $L'$, and the $L'$ passed to the sub-recursion for $v$ is $S_v \cup \dem(v)$ and $|\dem(v)| \leq s$.
	
	Inside each recursion of $\solve$, the number of nested conditioning operations is $|L'| + 1 \leq s(\mathrm{depth}(u) + 1) \leq s(h+2)$. Since the recursion can take up to $h+1$ levels, the number $R$ of rounds we need is at most $s(h+2)(h+1) + 1 = O(sh^2)$.
\end{proof}

%Now we start to bound the expected cost of $\tilde T$ and the probability that each label $\ell \in K$ appears in $\tilde T$. We need to set up some notations.  
\paragraph{Notations and Maintenance of Marginal Probabilities.}  We say that an event $e \in \bbD$ is inside $T^0[u]$ for some $u \in V(T^0)$ if either $e = v \in V(T^0[u])$ or $e = (v, \ell)$ for some $v \in V(T^0[u])$. For every integer $i \in [0, sh]$, let $x^{(u, i)}$ be the value of $x$ after the $i$-th iteration of Loop~\ref{State:first-loop} in the recursion $\solve(u, \cdot, \cdot)$. If this recursion does not exist, then let $x^{(u, i)}$ be the all-$0$ vector over $\bbD$; if this recursion exists but Loop~\ref{State:first-loop} terminates in less than $i$ iterations in the recursion, then let $x^{(u, i)}$ be the value of $x$ at the end of loop.  Notice that Loop~\ref{State:first-loop} terminates in at most $sh$ iterations from the proof of Claim~\ref{claim:rounds}. 
%
%For every $u \in V(T^0)$, we let $s_u$ to denote the random string used in the recursion of $\solve(u, \cdot, \cdot)$ that decides the random choices inside the recursion.  Let $x^{(u)}$ be the $x$ at the beginning of the recursion $\solve$ for $u$; if $u \notin \tilde T$, then we let $x^{(u)}$ be the all-0 vector with domain $\bbD$.  Notice that $x^{(u)}$ is a function of $(s_v)_{v\text{ is ancestor of }u}$ (and deterministic variables). 

The randomness of the algorithm comes from Steps~\ref{State:sampling} and \ref{State:second-loop} in $\solve$. Each time we run Step~\ref{State:sampling} or \ref{State:second-loop}, we assume we first generate a random number and then use it to make the decision.  We say a random number is generated before $x^{(u, i)}$, if the random number is generated in $\solve(u', \cdot, \cdot)$ for some ancestor $u'$ of $u$, or in $\solve(u, \cdot, \cdot)$ before or at the $i$-th iteration of Loop~\ref{State:first-loop}. Notice that each $x^{(u, i)}$ is completely determined by the random numbers generated before it. The following two claims state that the marginal probabilities of events are maintained in our random process.

%	For any root-to-leaf path $u_1 = u_1, u_2, \cdots, u_t$ in $T^0$, a consecutive sequence of variables is sub-interval of variables in the sequence $x^{(u_1, 0)}_e, x^{(u_1, 1)}_e, \cdots, x^{(u_1, sh)}_e, x^{(u_2, 0)}_e, x^{(u_2, 1)}_e, \cdots, x^{(u_2, sh)}_e, \cdots, x^{(u_t, 0)}_e, x^{(u_t, 1)}_e, \cdots, x^{(u_t, sh)}_e$.

\begin{claim}
	\label{claim:marginal-1}
	Let $u \in V(T^0), i \in [sh]$, $x^\old = x^{(u, i-1)}$ and $x^\new = x^{(u, i)}$. Let $\calE$ be any event determined by the random numbers generated before $x^\old = x^{(u, i-1)}$. Then, for every $e \in \bbD$, we have 
	\begin{align*}
	\E[x^\new_e|x^\old_e, \calE] = x^\old_e.
	\end{align*}
\end{claim}
\begin{proof}
	Conditioned on that the $i$-th iteration of $\solve(u, \cdot, \cdot)$ does not exist, the equality holds trivially. So we condition on that the iteration exists. Let $L'$ be the $L'$ passed to $\solve(u, \cdot, \cdot)$; then the $\ell$ handled in the $i$-th iteration is determined by $L'$ and $i$. So, 
	\begin{align*}
	\E\left[x^\new_e\big|x^\old, L'\right] = \sum_{v \in \Lambda_{u}}x^\old_{(v, \ell)}\cdot \frac{x^\old_{\{e, (v, \ell)\}}}{x^\old_{(v, \ell)}} = \sum_{v \in \Lambda_{u}} x^\old_{\{e, (v, \ell)\}} = x^\old_{\{e, (u, \ell)\}} = x^\old_e. 
	\end{align*}
	The first equality is by the random process for choosing $v$ and the definition of the conditioning operation. % and the assumption that $\calE$ does not depend on the random strings inside $(s_w)_{w \in T^0[u]}$. 
	The second-to-last equality follows from Constraint~\eqref{LPC:sum-over-children}, and the last equality follows from $x^\old_{(u, \ell)} = 1$ and Property~\ref{property:SA-if-event-happens}. 
	
	Also, given $x^\old$ and $L'$, the random process in the $i$-th iteration of $\solve(u, \cdot, \cdot)$ does not depend on the random numbers generated before $x^\old$, and thus does not depend on $\calE$.  Therefore, $\E\left[x^\new_e\big|x^\old, L', \calE\right] = x^\old_e$.  Deconditioning over $L'$ and the components inside $x^{\old}$ other than $x^\old_e$ gives $\E\left[x^\new_e\big|x^\old_e, \calE\right] = x^\old_e$. 
\end{proof}

\begin{claim}
	\label{claim:marginal-2}
	Let $u \in V^\internal, v \in \Lambda_u, x^\old = x^{(u, sh)}$ and $x^\new = x^{(v, 0)}$. Let $\calE$ be any event determined by the random numbers generated before $x^\old = x^{(u, sh)}$. Then, for any event $e$ inside $T^0[v]$, 
	we have
	\begin{align*}
	\E\left[x^\new_e\big|x^\old_e, \calE\right] = x^\old_e.
	\end{align*}
\end{claim}
\begin{proof}
	Again we can condition on that the recursion $\solve(u, \cdot, \cdot)$ exists. Consider the iteration of Loop~\ref{State:second-loop} for $v$ in $\solve(u, \cdot, \cdot)$. We have
	\begin{align*}
	\E\left[x^\new_e\big| x^\old, \calE \right] = x^\old_v \times \frac{x^\old_{\{e, v\}}}{x^\old_v} = x^\old_{\{e, v\}} = x^\old_e. 
	\end{align*}
	The first equality holds since we make the recursive call for $v$ with probability $x^\old_v$; given $x^\old$, this is independent of $\calE$. The last equality comes from that event $e$ is inside $T^0[v]$ and thus $\hat x_e \leq \hat x_v$ for every $\hat x \in \calP$; Property~\ref{property:SA-implication} gives the equality.

%	Multiplying both sides by $x_e$ and linearizing the constraint gives $x_e \leq x_{\{v, e\}}$.  But $x_{\{v, e\}} \leq x_e$; thus $x_e = x_{\{v, e\}}$ is implied.
	 Again, deconditioning over the components inside $x^{\old}$ other than $x^\old_e$ gives $\E\left[x^\new_e\big|x^\old_e, \calE\right] = x^\old_e$. 
\end{proof}

\begin{coro}
	\label{coro:marginal-of-a-node}
	For every $v \in V(T^0)$, we have $\Pr[v \in \tilde V] = x^*_v$.
\end{coro}
\begin{proof}
	Let $u_1 = {\root(T^0)}, u_2, \cdots, u_t = v$ be the path from ${\root(T^0)}$ to $v$ in $T^0$. Applying Claims~\ref{claim:marginal-1} and \ref{claim:marginal-2}, we can obtain that the sequence $x^{(u_1, 0)}_v$, $x^{(u_1, 1)}_v$, $\cdots$, $x^{(u_1, sh)}_v$, $x^{(u_2, 0)}_v$, $x^{(u_2, 1)}_v$, $\cdots$, $x^{(u_2, sh)}_v$, $\cdots$, $x^{(u_{t-1}, 0)}_v$, $x^{(u_{t-1}, 1)}_v$, $\cdots$, $x^{(u_{t-1}, sh)}_v$, $x^{(u_t, 0)}_v$ forms a martingale. This holds since all variables before a variable $x^{(u', i)}$ in the sequence are determined only by random numbers generated before $x^{(u', i)}$.
	Thus $\Pr[v \in \tilde V] = \E\left[x^{(v, 0)}_v\right] = x^{({\root(T^0)}, 0)} = x^*_v$ as $x^{({\root(T^0)}, 0)}_v$ is deterministic.
\end{proof}

Then it is immediately true that the expected cost of $\tilde T$ is small.
\begin{coro}
	$\E[\cost(\tilde T)] \leq \opt$.
\end{coro}
\begin{proof}
	$\E[\cost(\tilde T)] = \sum_{v \in V(T^0)} \Pr[v \in \tilde V]\cdot c_v = \sum_{v \in V(T^0)} x^*_v  c_v \leq \opt$.
\end{proof}

\paragraph{Bounding Probability of Label $\ell \in K$ Appearing in $\tilde T$} To finish the proof of Theorem~\ref{thm:find-tree}, it suffices to show that the probability that a label $\ell \in K$ is provided by $\tilde T$ with high probability. \emph{Till the end of the proof, we shall fix a label $\ell \in K$.}  %Notice that if $\ell$ appears in $\tilde T$, then it must appear in $\tilde V \cap V^\leaf$ since \Cref{claim:label-consistent} says the tree $\tilde T$ is label-consistent. 

Let $t_\ell =\big|\{v \in \tilde V \cap V^\leaf: a_v = \ell\}\big|$ be the number of nodes in $\tilde V \cap V^\leaf$ with label $\ell$. Our goal is to prove that $t_\ell \geq 1$ with high probability.  The proof is almost the same as the counterpart in \cite{Rothvoss11}; we include it here for completeness.
\begin{lemma}
	\label{lemma:E-t-ell}
	$\E\left[t_\ell\right] = 1$.
\end{lemma}
\begin{proof}
	By Corollary~\ref{coro:marginal-of-a-node}, we have $$\E[t_\ell] = \E\left[\big|\{v \in \tilde V \cap V^\leaf: a_v = \ell\}\big|\right] =\sum_{v \in V^\leaf: a_v = \ell} \Pr[v \in \tilde V] = \sum_{v \in V^\leaf: a_v = \ell} x^*_v  = x^*_{({\root(T^0)}, \ell)}= 1,$$ where the second-to-last equality follows from Claim~\ref{claim:sum-up-leaf-decedants}, and the last equality is by \eqref{LPC:root}.
\end{proof}

\begin{lemma}
	\label{lemma:condition-on-one-node}
	For every $w \in V^\leaf$ with $a_w = \ell$, we have $\E[t_\ell|w \in {\tilde V}] \leq h + 1$.
\end{lemma}
\begin{proof}
	Assume $w$ is at depth $h'$ in the tree ${T^0}$.  We partition the set $\{w' \in V^\leaf \setminus \{w\}: a_{w'} = \ell\}$ of leaves into $h'$ sets $U_0, U_1, \cdots, U_{h'-1}$ according to the LCA of $w'$ and $w$: $w'$ is in $U_i$ if the LCA of $w'$ and $w$ has depth $i$ in the tree ${T^0}$ (the root ${\root(T^0)}$ has depth $0$).  Notice that $w' \neq w$ and thus the LCA has depth between 0 and $h'-1$. We show that for every $i = 0, 1, \cdots, h'-1$,
	\begin{align}
	\E\left[|U_i \cap {\tilde V}| \big| w \in {\tilde V}\right] \leq 1.  \label{inequ:each-level}
	\end{align}
	Summing up the inequality over all $i = 0, 1, \cdots, h'-1$ and taking $w$ itself into account implies $\E[t_\ell|w \in {\tilde V}] \leq h' + 1 \leq h + 1$.
	
	Thus, it remains to prove \eqref{inequ:each-level}. We fix an $i \in \{0, 1,\cdots, h'-1\}$ and let $u$ be the ancestor of $w$ with depth $i$.  Focus on any $w' \in U_i$; thus $u$ is the LCA of $w'$ and $w$.  %If $w \in {\tilde V}$, then we must have called $\solve$ for $u$.  Then 
	Let $(S_v)_{v \in \Lambda_u}$ be the vector $(S_v)_{v \in \Lambda_u}$ before Loop~\ref{State:second-loop} in $\solve(u, \cdot, \cdot)$. 
	
	Given $\{S_v\}_{v \in \Lambda_u}$ and $x^{(u, sh)}$, the two events $w \in \tilde V$ and $w' \in \tilde V$ are independent. Thus, 
	\begin{align*}
	\Pr\left[w' \in \tilde V\big |\{S_v\}_{v \in \Lambda_u}, x^{(u, sh)}, w \in \tilde V\right] &= \Pr\left[w' \in \tilde V \big|\{S_v\}_{v \in \Lambda_u}, x^{(u, sh)}\right] \\
	&= \E\left[x^{(w', 0)}_{w'} \big|\{S_v\}_{v \in \Lambda_u}, x^{(u, sh)}\right]= x^{(u, sh)}_{w'}.
	\end{align*}
	To see the third equality, consider the path $u, u_1, u_2, \cdots, u_t = w'$ from $u$ to $w'$ in $T^0$. Then Claims~\ref{claim:marginal-1} and \ref{claim:marginal-2} imply that conditioned on $\{S_v\}_{v \in \Lambda_u}$ and $x^{(u, sh)}$, the sequence $x^{(u_1, 0)}, \break x^{(u_1, 1)}, \cdots, x^{(u_1, sh)}, x^{(u_2, 0)}, x^{(u_2, 1)}\cdots, x^{(u_{t-1}, sh)}, x^{(u_t,0)}$ is a martingale.
	
	Summing up over all $w' \in U_i$, we have
	\begin{align*}
	\E\left[|U_i \cap {\tilde V}|\big| \{S_v\}_{v \in \Lambda_u}, x^{(u, sh)}, w \in \tilde V\right] = \sum_{w' \in U_i}x^{(u, sh)}_{w'} = \sum_{w' \in U_i}x^{(u, sh)}_{(w', \ell)}  \leq x^{(u, sh)}_{(u, \ell)} \leq 1. 
	\end{align*}
	The first inequality used Claim~\ref{claim:sum-up-leaf-decedants} and $U_i \subseteq \Lambda^\leaf_u$.
	Deconditioning gives \eqref{inequ:each-level}. 		
\end{proof}

\begin{lemma}
	\label{lemma:conditional-E-t-ell}
	For every $\ell \in K$, we have $E[t_\ell|t_\ell \geq 1] \leq h+1$.
\end{lemma}	
\begin{proof}
	In the following, $w$ and $w'$ in summations are over all nodes in $V^\leaf$ with label $\ell$. 
	%For every $w \in V^\leaf$ with $a_w = \ell$, let $z_w \in \{0, 1\}$ indicate whether $w \in \tilde V$ or not.  
	%Then $t_\ell = \sum_{w \in V^\leaf: a_w = \ell}z_w$, and Lemma xx says $\E[t_\ell|z_w = 1] \leq h+1$. So, we have
	\begin{align*}
	\E[t_\ell | t_\ell \geq 1]^2 &\leq \E[t_\ell^2|t_\ell \geq 1] = \sum_{w, w'} \Pr[w \in \tilde V, w' \in \tilde V | t_\ell \geq 1] \\
	& \hspace*{0.28\textwidth} \text{(by Jansen's inequality and the definition of $t_\ell$)}\\
	& = \sum_{w} \Pr[w \in \tilde V|t_\ell \geq 1] \sum_{w'} \Pr[w' \in \tilde V|w \in \tilde V, t_\ell \geq 1] \\
	&= \sum_{w}\Pr[w \in \tilde V|t_\ell \geq 1]\E[t_\ell|w \in \tilde V] \\
	&\hspace*{0.22\textwidth}  \text{(by the definition of $t_\ell$ and that $w \in \tilde V$ implies $t_\ell \geq 1$)}\\
	& \leq (h+1)\sum_{w} \Pr[w \in \tilde V|t_\ell \geq 1] \hspace*{0.3\textwidth} \text{(by Lemma~\ref{lemma:condition-on-one-node})}\\
	&=  (h+1)\E[t_\ell|t_\ell \geq 1] \hspace*{0.3\textwidth} \text{(by the definition of $t_\ell$).}
	\end{align*}
	This implies $\E[t_\ell|t_\ell \geq 1] \leq h+1$.
\end{proof}

\begin{coro}
	$\Pr[t_\ell \geq 1 ] \geq \frac{1}{h+1}$ for every $\ell \in K$.
\end{coro}
\begin{proof}
	Notice that $1 = \E[t_\ell] = \E[t_\ell | t_\ell \geq 1] \cdot \Pr[t_\ell \geq 1]$. The corollary follows from Lemma~\ref{lemma:conditional-E-t-ell}.
\end{proof}
Thus we have finished the proof of Theorem~\ref{thm:find-tree}.

\section{Discussion and Open Problems}
\label{sec:discussion}

\fabr{Moved up the last sentences: better flow}
In this paper we close the gap on the approximability of DST for the class of quasi-polynomial-time algorithms. %This \hg{gap} was closed by Cygan~et~al.~\cite{CyganKL17-unpublished} for the class of mildly sub-exponential-time algorithms.\fabr{Not sure I get this sentence: do Cygan et al.  consider exact algorithms?}  
However, there is still a \fab{huge} gap between the lower and upper bounds on approximation ratios for the class of polynomial-time algorithms. In particular, it has been an open problem that perplex\bu{es} many researchers whether DST admits a polylogarithmic approximation algorithm that runs in polynomial-time. There are both positive and negative evidence\bu{s} that suggest DST may or may not admit such algorithm. \fab{On one hand}, Rothvo{\ss}~\cite{Rothvoss11} observes that despite an algorithm based on hierarchical techniques (i.e., Sum-of-Square\fab{s}) runs in super polynomial-time due to the size of the lifted linear program, the rounding algorithm itself reads only a polynomial \fab{number of variables} of the fractional solution with high probability. This also applies to all the LP techniques including the folklore path-tree formulation (please see, e.g., \cite{Laekhanukit16}). Thus, some may believe that DST admits polylogarithmic approximation algorithms that run in polynomial-time. On the other han\fab{d}, the factor $ n^{\epsilon}/\epsilon$ that appears in the approximation ratio shows that same behavior as \fab{in} other problems whose trade-off between approximation ratio and running-time are tight under \fab{the} Exponential-Time Hypothesis, e.g., {\em Dense CSP} \cite{ManurangsiR17} and {\em Densest $k$-Subgraph} \cite{Manurangsi17}\footnote{In \cite{ManurangsiR17}, the trade-off is slightly weaker, say  $O(n^{\epsilon^3}/\epsilon)$-approximation ratio versus $n^{1/\epsilon}$-running time.}
Our result removes the factor $1/\epsilon$ from the approximation ratio, suggesting that DST may have a different behavior tha\fab{n} the other problems mentioned above. 
Nevertheless, our technique does not yield a good trade-off between approximation ratio and running-time as it requires exactly quasi-polynomial-time to remove such factor. It seems that there is still a \fab{major} barrier in answering the open question. 
%%%%
%In fact, it might be possible that DST admits a polynomial-time {\em estimation}\footnote{An {\em $\alpha$-estimation} algorithm is an algorithm that estimates the value (cost) of the optimal solution to within a factor of $\alpha$ but does output any feasible solution. The concept of estimation versus approximation was shown in \cite{FeigeJ15}.} algorithm that yields polylogarithmic factor, while finding or approximating the solution requires at least quasi-polynomial-time. {\color{red} Shi: is there an evidence to support this?}
%Now as we close the gap for the class of quasi-polynomial-time algorithms, and the gap for the class of mildly sub-exponential-time algorithm was closed by Cygan~et~al.~\cite{XXX}. \fabr{\hg{missing citation}} The only remaining question is to close the gap between the lower and upper bounds on approximation ratios for the class of polynomial-time algorithms. 

\paragraph{Acknowledgement.}  We would like to thank Uriel Feige for useful discussions over two years, and we would like to thank Jittat Fakcharoenphol for useful discussion on the balanced tree separator.\fabr{Moved grants back} 

F. Grandoni is partially supported by the SNSF Grant 200021\_159697/1 and the SNSF Excellence Grant 200020B\_182865/1.

B. Laekhanukit is supported by the National 1000-Youth Award by the Chinese government. Parts of this work was done when Laekhanukit was at the Weizmann Institute of Science, partially supported by ISF grant \#621/12 and I-CORE grant \#4/11, while he was visiting the Simons Institute for the Theory of Computing, which was partially supported by the DIMACS/Simons Collaboration on Bridging Continuous and Discrete Optimization through NSF grant \#CCF-1740425, and while he was at the Max-Plack Institute for Informatics.

S. Li is supported by NSF grant \#CCF-1566356 and \#CCF-1717134. Some critical parts of this work were done while Li was visiting the Institute for Theoretical Computer Science at Shanghai University of Finance and Economics.

%Bundit Laekhanukit was supported by the National 1000-Youth Award by the Chinese government,
%and Shi Li was supported by NSF grant \#CCF-1566356 and \#CCF-1717134.

%Parts of this work were done while Bundit Laekhanukit were at the Weizmann Institute of Science, which was partially supported by ISF grant \#621/12 and I-CORE grant \#4/11, while he was visiting the Simons Institute for the Theory of Computing, which was partially supported by the DIMACS/Simons Collaboration on Bridging Continuous and Discrete Optimization through NSF grant \#CCF-1740425, and while he was at the Max-Plack Institute for Informatics.
%%
%The critical parts of this work were done while Shi Li was visiting the Institute for Theoretical Computer Science at Shanghai University of Finance and Economics.
%
%
%Bundit Laekhanukit was supported by the National 1000-Youth Award by the Chinese government,
%and Shi Li was supported by NSF grant CCF-1566356 and CCF-1717134.
%
%Part of this work was done while Bundit Laekhanukit was visiting the Simons Institute for the Theory of Computing. It was partially supported by the DIMACS/Simons Collaboration on Bridging Continuous and Discrete Optimization through NSF grant \#CCF-1740425.
%
%This work was also partially done while Bundit Laekhanukit was at the Max-Planck Institute for Informatics, and the crucial parts were done while Shi Li was visiting the Institute for Theoretical Computer Science at Shanghai University of Finance and Economics.

%\bibliographystyle{alpha}
\bibliographystyle{plain}
\bibliography{dst}

\begin{thebibliography}{10}

\bibitem{Bartal96}
Yair Bartal.
\newblock Probabilistic approximations of metric spaces and its algorithmic
  applications.
\newblock In {\em 37th Annual Symposium on Foundations of Computer Science,
  {FOCS} '96, Burlington, Vermont, USA, 14-16 October, 1996}, pages 184--193,
  1996.

\bibitem{BCG09}
MohammadHossein Bateni, Moses Charikar, and Venkatesan Guruswami.
\newblock Maxmin allocation via degree lower-bounded arborescences.
\newblock In {\em Proceedings of the 41st Annual {ACM} Symposium on Theory of
  Computing, {STOC} 2009, Bethesda, MD, USA, May 31 - June 2, 2009}, pages
  543--552, 2009.

\bibitem{ByrkaGRS13}
Jaroslaw Byrka, Fabrizio Grandoni, Thomas Rothvo{\ss}, and Laura Sanit{\`{a}}.
\newblock Steiner tree approximation via iterative randomized rounding.
\newblock {\em J. {ACM}}, 60(1):6:1--6:33, 2013.

\bibitem{ChalermsookGL15}
Parinya Chalermsook, Fabrizio Grandoni, and Bundit Laekhanukit.
\newblock On survivable set connectivity.
\newblock In {\em SODA}, pages 25--36, 2015.

\bibitem{CharikarCCDGGL99}
Moses Charikar, Chandra Chekuri, To{-}Yat Cheung, Zuo Dai, Ashish Goel, Sudipto
  Guha, and Ming Li.
\newblock Approximation algorithms for directed steiner problems.
\newblock {\em J. Algorithms}, 33(1):73--91, 1999.

\bibitem{ChekuriP05}
Chandra Chekuri and Martin P{\'{a}}l.
\newblock A recursive greedy algorithm for walks in directed graphs.
\newblock In {\em 46th Annual {IEEE} Symposium on Foundations of Computer
  Science {(FOCS} 2005), 23-25 October 2005, Pittsburgh, PA, USA, Proceedings},
  pages 245--253, 2005.

\bibitem{CheriyanLNV14}
Joseph Cheriyan, Bundit Laekhanukit, Guyslain Naves, and Adrian Vetta.
\newblock Approximating rooted steiner networks.
\newblock {\em {ACM} Transactions on Algorithms}, 11(2):8:1--8:22, 2014.

\bibitem{Chlamtac07}
Eden Chlamtac.
\newblock Approximation algorithms using hierarchies of semidefinite
  programming relaxations.
\newblock In {\em 48th Annual {IEEE} Symposium on Foundations of Computer
  Science {(FOCS} 2007), October 20-23, 2007, Providence, RI, USA,
  Proceedings}, pages 691--701, 2007.

\bibitem{CGM13}
Marek Cygan, Fabrizio Grandoni, and Monaldo Mastrolilli.
\newblock How to sell hyperedges: The hypermatching assignment problem.
\newblock In {\em Proceedings of the Twenty-Fourth Annual {ACM-SIAM} Symposium
  on Discrete Algorithms, {SODA} 2013, New Orleans, Louisiana, USA, January
  6-8, 2013}, pages 342--351, 2013.

\bibitem{EneCKP15}
Alina Ene, Deeparnab Chakrabarty, Ravishankar Krishnaswamy, and Debmalya
  Panigrahi.
\newblock Online buy-at-bulk network design.
\newblock In Venkatesan Guruswami, editor, {\em {IEEE} 56th Annual Symposium on
  Foundations of Computer Science, {FOCS} 2015, Berkeley, CA, USA, 17-20
  October, 2015}, pages 545--562. {IEEE} Computer Society, 2015.

\bibitem{FakcharoenpholRT04}
Jittat Fakcharoenphol, Satish Rao, and Kunal Talwar.
\newblock A tight bound on approximating arbitrary metrics by tree metrics.
\newblock {\em J. Comput. Syst. Sci.}, 69(3):485--497, 2004.

\bibitem{FriggstadKKLST14}
Zachary Friggstad, Jochen K{\"{o}}nemann, Young Kun{-}Ko, Anand Louis, Mohammad
  Shadravan, and Madhur Tulsiani.
\newblock Linear programming hierarchies suffice for directed steiner tree.
\newblock In {\em Integer Programming and Combinatorial Optimization - 17th
  International Conference, {IPCO} 2014, Bonn, Germany, June 23-25, 2014.
  Proceedings}, pages 285--296, 2014.

\bibitem{GargKR00}
Naveen Garg, Goran Konjevod, and R.~Ravi.
\newblock A polylogarithmic approximation algorithm for the group steiner tree
  problem.
\newblock {\em J. Algorithms}, 37(1):66--84, 2000.

\bibitem{GKL19}
Shashwat Garg, Janardhan Kulkarni, and Shi Li.
\newblock Lift and project algorithms for precedence constrained scheduling to
  minimize completion time.
\newblock In {\em Proceedings of the Thirtieth Annual {ACM-SIAM} Symposium on
  Discrete Algorithms, {SODA} 2019, New Orleans, Louisiana, USA, January 6-8,
  2019}.

\bibitem{GrandoniL17}
Fabrizio Grandoni and Bundit Laekhanukit.
\newblock Surviving in directed graphs: a quasi-polynomial-time polylogarithmic
  approximation for two-connected directed steiner tree.
\newblock In Hatami et~al. \cite{DBLP:conf/stoc/2017}, pages 420--428.

\bibitem{GuptaKR10}
Anupam Gupta, Ravishankar Krishnaswamy, and R.~Ravi.
\newblock Tree embeddings for two-edge-connected network design.
\newblock In {\em Proceedings of the Twenty-First Annual {ACM-SIAM} Symposium
  on Discrete Algorithms, {SODA} 2010, Austin, Texas, USA, January 17-19,
  2010}, pages 1521--1538, 2010.

\bibitem{HalperinK03}
Eran Halperin and Robert Krauthgamer.
\newblock Polylogarithmic inapproximability.
\newblock In Lawrence~L. Larmore and Michel~X. Goemans, editors, {\em
  Proceedings of the 35th Annual {ACM} Symposium on Theory of Computing, June
  9-11, 2003, San Diego, CA, {USA}}, pages 585--594. {ACM}, 2003.

\bibitem{DBLP:conf/stoc/2017}
Hamed Hatami, Pierre McKenzie, and Valerie King, editors.
\newblock {\em Proceedings of the 49th Annual {ACM} {SIGACT} Symposium on
  Theory of Computing, {STOC} 2017, Montreal, QC, Canada, June 19-23, 2017}.
  {ACM}, 2017.

\bibitem{HelvigRZ01}
Christopher~S. Helvig, Gabriel Robins, and Alexander Zelikovsky.
\newblock An improved approximation scheme for the group steiner problem.
\newblock {\em Networks}, 37(1):8--20, 2001.

\bibitem{Jordan1869}
Camille Jordan.
\newblock Sur les assemblages de lignes.
\newblock {\em Journal für die reine und angewandte Mathematik}, 70:185--190,
  1869.

\bibitem{KKN12}
Rohit Khandekar, Guy Kortsarz, and Zeev Nutov.
\newblock Approximating fault-tolerant group-steiner problems.
\newblock {\em Theorerical Computer Science}, 416:55--64, 2012.

\bibitem{Kortsarz01}
Guy Kortsarz.
\newblock On the hardness of approximating spanners.
\newblock {\em Algorithmica}, 30(3):432--450, 2001.

\bibitem{Laekhanukit14}
Bundit Laekhanukit.
\newblock Parameters of two-prover-one-round game and the hardness of
  connectivity problems.
\newblock In {\em SODA}, pages 1626--1643, 2014.

\bibitem{Laekhanukit16}
Bundit Laekhanukit.
\newblock Approximating directed steiner problems via tree embedding.
\newblock In Ioannis Chatzigiannakis, Michael Mitzenmacher, Yuval Rabani, and
  Davide Sangiorgi, editors, {\em 43rd International Colloquium on Automata,
  Languages, and Programming, {ICALP} 2016, July 11-15, 2016, Rome, Italy},
  volume~55 of {\em LIPIcs}, pages 74:1--74:13. Schloss Dagstuhl -
  Leibniz-Zentrum fuer Informatik, 2016.

\bibitem{LR16}
Elaine Levey and Thomas Rothvoss.
\newblock A (1+epsilon)-approximation for makespan scheduling with precedence
  constraints using {LP} hierarchies.
\newblock In {\em Proceedings of the 48th Annual {ACM} {SIGACT} Symposium on
  Theory of Computing, {STOC} 2016, Cambridge, MA, USA, June 18-21, 2016},
  pages 168--177, 2016.

\bibitem{Manurangsi17}
Pasin Manurangsi.
\newblock Almost-polynomial ratio eth-hardness of approximating densest
  k-subgraph.
\newblock In Hatami et~al. \cite{DBLP:conf/stoc/2017}, pages 954--961.

\bibitem{ManurangsiR17}
Pasin Manurangsi and Prasad Raghavendra.
\newblock A birthday repetition theorem and complexity of approximating dense
  csps.
\newblock In Ioannis Chatzigiannakis, Piotr Indyk, Fabian Kuhn, and Anca
  Muscholl, editors, {\em 44th International Colloquium on Automata, Languages,
  and Programming, {ICALP} 2017, July 10-14, 2017, Warsaw, Poland}, volume~80
  of {\em LIPIcs}, pages 78:1--78:15. Schloss Dagstuhl - Leibniz-Zentrum fuer
  Informatik, 2017.

\bibitem{Moshkovitz15}
Dana Moshkovitz.
\newblock The projection games conjecture and the np-hardness of ln
  n-approximating set-cover.
\newblock {\em Theory of Computing}, 11:221--235, 2015.

\bibitem{RobinsZ05}
Gabriel Robins and Alexander Zelikovsky.
\newblock Tighter bounds for graph steiner tree approximation.
\newblock {\em {SIAM} J. Discrete Math.}, 19(1):122--134, 2005.

\bibitem{Rothvoss11}
Thomas Rothvo{\ss}.
\newblock Directed steiner tree and the lasserre hierarchy.
\newblock {\em CoRR}, abs/1111.5473, 2011.

\bibitem{Zelikovsky93}
Alexander Zelikovsky.
\newblock An 11/6-approximation algorithm for the network steiner problem.
\newblock {\em Algorithmica}, 9(5):463--470, 1993.

\bibitem{Zelikovsky97}
Alexander Zelikovsky.
\newblock A series of approximation algorithms for the acyclic directed steiner
  tree problem.
\newblock {\em Algorithmica}, 18(1):99--110, 1997.

\bibitem{ZosinK02}
Leonid Zosin and Samir Khuller.
\newblock On directed steiner trees.
\newblock In David Eppstein, editor, {\em Proceedings of the Thirteenth Annual
  {ACM-SIAM} Symposium on Discrete Algorithms, January 6-8, 2002, San
  Francisco, CA, {USA.}}, pages 59--63. {ACM/SIAM}, 2002.

\end{thebibliography}

\newpage
\appendix

\section{The Missing Proofs from \Cref{sec:prelim}}
\label{apendix:balcned-tree-partition}

We provide in the section the proof of \Cref{cor:balanced-partition}.

\balancedpartition*

This can be proved via the well-known Tree-Separator Theorem:
 
 \begin{theorem}[Tree-Separator Theorem \cite{Jordan1869}]
 \label{thm:tree-separator}
 For any $n$-vertex tree $T$, there is a vertex $v\in V(T)$ such that 
 removing $v$ from $T$ results in a graph where each (connected)
 component contains at most $n/2$ vertices.
 \end{theorem}

\begin{proof}[Proof of \Cref{cor:balanced-partition}]
	The proof follows straightforward from \Cref{thm:tree-separator}. We may assume that $T$ is an out-arborescence rooted at a vertex $r$.
	We pick a vertex $v$ as in \Cref{thm:tree-separator} (it could be the case that $v = r$).
	Then we have weakly connected subgraphs of $T\setminus{v}$, say $H_1,\ldots,H_q$.
	It is not hard to see that each subgraphs $H_i$, for $i\in[q]$, is 
	an arborescence. %, and only one of them has $r$ as a root, say $H_1$.
	
	We start from $T' = \emptyset$. As long as there is an $H_i \not\subseteq T'$ such that $|V(H_i)| + |V(T')| < 2n/3$, we add $H_i$ to $T'$.
	Let $T_1$ contain $v$, $T'$ and the edges between $v$ and $T'$; let $T_2$ contain $v$, sub-graphs $H_i$ that are not in $T'$, and the edges joining $v$ and these sub-graphs. 
	
	It follows from the construction that both $T_1$ and $T_2$ are sub-arborescences of $T$
	that have only $v$ as a common vertex and that $T_1\cup T_2 = T$.  Renaming $T_1$ and $T_2$ in the end of the proof if necessary so that $T_1$ is rooted at $r$.
	Notice that $|V(T')| < 2n/3$, implying that $V(T_1) < 2n/3 + 1$.
	It is sufficient to show that $|V(T')| > n/3 - 1$, which will imply $V(T_2) < 2n/3+1$ since $|V(T')|+|V(T_2)| = n$.
	
	Suppose $|V(T')| < n/3 - 1$. 
	Then every component $H_i$ not included in $T'$
	must contain more than $n/3 + 1$ vertices.
	So there are at most two such components. Also, there can not be just one such component since otherwise it has size more than $2n/3 > n/2$, a contradiction. 
	%If $T_2$ contains only one component $H_i$, then the claim, again, follows immediately because $|V(H_i)\cup \{v\}|\leq 2n/3+1$, for all $i\in[q]$.
	So, there are exactly two components not in $T'$, and one of them, say $H_i$, has at most $\frac{n-1-|V(T')|}{2}$ vertices. But then $|V(T')| + |V(H_i)| \leq \frac{n-1+|V(T')|}{2} \leq 2n/3 - 1 < 2n/3$, a contradiction. 
\end{proof}

\claimSA*
\begin{proof}
	Let $x \in \SA(\calP, R)$ for some $R \geq 2$. 
	\begin{enumerate}[label=(\ref{claim:SA}\alph*),leftmargin=*]
		\item Consider the case where $S' = S \cup \{i\}$ for some $i \notin S$.  Linearizing the constraint $x_i\leq 1$ multiplied by $\sum_{i' \in S}x_i$ gives the constraint $x_{S'} \leq x_{S}$. %This implies, all the variables have values in $[0, 1]$ as $x_\emptyset = 1$.
		\item Multiplying $1-x_{i}\geq 0$ and  $1-x_{i'} \geq 0$ and linearizing the product gives the constraint $1 - x_{i} - x_{i'} + x_{\{i, {i'}\}} \geq 0$. Then $x_{i} = 1$ implies $x_{i'} \leq x_{\{i, {i'}\}}$. But $x_{i'} \geq x_{\{i, {i'}\}}$; thus $x_{i'} = x_{\{i, {i'}\}}$.
		\item $x_i \leq x_{i'}$ is implied by the constraints for the basic polytope.  Multiplying both sides by $x_i$ and linearizing the constraint gives $x_{i} \leq x_{\{i, i'\}}$; but $x_{i} \geq x_{\{i, i'\}}$ by (\ref{claim:SA}a). Thus $x_{i} = x_{\{i, i'\}}$.
%		\item The constraint ${\sum}_{i' = 1}^n a_{i'} x_{i'} \leq b$ is implied by constraints for $\calP$; multiplying the constraint with $x_i$ gives the property.
	\end{enumerate}
	Now, let $x'$ be obtained from $x$ by conditioning on some event $i \in [n]$.
	\begin{enumerate}[label=(\ref{claim:SA}\alph*),leftmargin=*,start = 4]
		\item By definition of the conditioning operation, we have $x'_i = \frac{x_{\{i\} \cup \{i\}}}{x_i} = \frac{x_i}{x_i} = 1$. 
		\item $x'_\emptyset  = \frac{x_{\emptyset \cup \{i\}}}{x_i} = \frac{x_i}{x_i} = 1$. \eqref{inequ:SA} on $x'$ for $j, S$ and $T$ is implied by \eqref{inequ:SA} on $x$ for $j, S \cup \{i\}$ and $T$.
		\item If $x_{i'} = 0$, then $x'_{i'} = \frac{x_{\{i',i\}}}{x_i} = 0$ since $x_{\{i',i\}} \leq x_{i'} = 0$. Consider the case $x_{i'} = 1$. Property~\ref{property:SA-if-event-happens} says $x_{\{i, i'\}} = x_i$, implying $x'_{i'} = \frac{x_{\{i,i'\}}}{x_i} = 1$. \qedhere
	\end{enumerate}
\end{proof}

\section{Missing Proofs from \Cref{sec:reduction}}

	\claimbinarydecompositiontree*
	\begin{proof}
		Clearly, $\tau^*$ is a full binary tree. It has depth $O(\log k)$ since $|V(T^*)| \leq 2k$ and the size of $|V(T)|$ goes down by a constant factor in each level of the recursion for $\codt$.  It is easy to see that $(e_{\beta})_{\beta\text{ is leaf of }\tau^*}$ is a 1-to-1 mapping from leaves of $\tau^*$ to $E(T^*)$, where an edge $(u, v) \in E(T^*)$ corresponds to a leaf $\beta$ of $\tau^*$ with $e_\beta = (u, v)$. This holds as $E(T_1)$ and $E(T_2)$ produced in Step~\ref{State:balance-partition} form a partition of $E(T)$, and in Step~\ref{State:create-beta} the leaf-node $\beta$ created has $e_\beta = (u, v)$. Thus,  $\cost(\tau^*) = \sum_{\beta \text{ a leaf of } \tau^* }c(e_\beta) = \sum_{e \in E(T^*)}c(e)= \opt$. Since every terminal $v \in K$ has in-degree exactly 1 in $T^*$, there is exactly one leaf $\beta \in V(\tau^*)$ with $\tail(e_\beta) = v$. In particular, all terminals in $K$ are involved in $\tau^*$. 
		
		A simple observation is that any tree $\tau^*$ returned by the procedure $\codt(T)$ will have $\mu_{\root(\tau)} = \root(T)$. Properties~\ref{property:root-of-decomposition-tree} and \ref{property:mu-is-head-of-e} hold trivially. Thus, to show that $\tau^*$ is indeed a decomposition tree, it suffices to prove Property~\ref{property:contain-root}. 
		
		Focus on a node $\alpha$ created in Step~\ref{State:create-alpha} in some recursion of $\codt$; we shall prove Property~\ref{property:contain-root} for this $\alpha$.  %If $T$ contains a single edge, then $\alpha$ has only one child $\beta$ with $\mu_\beta = \mu_{\alpha}$ and the property holds.  So we assume $T$ contains at least two edges. 
		Focus on the moment before we return the tree in  Step~\ref{State:return-decomp-tree} in the recursion. Let $\alpha_1 = \root(\tau_1)$ and $\alpha_2 = \root(\tau_2)$ be the two children of $\alpha$.  Then we have $\mu_{\alpha_1} = \root(T_1) = \root(T) = \mu_\alpha$ and $\mu_{\alpha_2} = \root(T_2)$. If $\root(T_2) = \root(T)$, then $\mu_{\alpha_2} = \mu_\alpha$ and there is nothing to prove. Thus, we assume $\root(T_2) \neq \root(T)$. Then $\root(T_2) \in V(T_1)$, and it has exactly one incoming edge in $T_1$. By our construction, there will be a leaf node $\beta \in V(\tau_1)$ with $e_\beta$ being the edge and thus $\tail(e_\beta) = \root(T_2)$. So $\mu_{\alpha_2} = \root(T_2)$ is involved in $\tau_1$.  Thus, Property~\ref{property:contain-root} holds.
	\end{proof} 
	
	\lemmadecompositiontoSteiner*
	\begin{proof}
		We simply let $T$ contain the edges $e_\alpha$ for all leaves $\alpha$ of $\tau$. Then the cost of $T$ is exactly $\cost(\tau)$. We shall show that $T$ contains a path from $r$ to every terminal $v \in K$.  At the end of the proof, we can remove edges in $T$ so that $T$ form\fab{s} an out-arborescence rooted at $r$.
		
		Given a node $\alpha$ of $\tau$, let $H_\alpha := (V, \{e_\beta: \beta\text{ is a leaf of } \tau[\alpha]\})$. We will show the following: 
	    \begin{equation}\tag{*}\label{eq:dec-induction}
	    \text{For every $\alpha \in V(\tau)$, $H_\alpha$ contains a path from $\mu_\alpha$ to every vertex $v$ involved in $\tau[\alpha]$.}  
	    \end{equation}
	    Since $E(T) = E(H_{\root(\tau)}), \mu_{\root(\tau)} = r$ and every terminal $v \in K$ is involved in $\tau$, applying \eqref{eq:dec-induction} to $\root(\tau)$ gives that $T$ contains a path from $r$ to every terminal in $K$, which finishes our proof.
		
		We prove \eqref{eq:dec-induction} by induction from the bottom to the top of the tree $\tau$. If $\alpha$ is a leaf, then $H_{\alpha}$ contains the edge $e_{\alpha}$, only $\head(e_{\alpha})$ and $\tail(e_{\alpha})$ are involved, and $\mu_\alpha = \head(e_\alpha)$. Thus, \eqref{eq:dec-induction} holds.  
		
		Now consider an internal node $\alpha$ in $\tau$, and assume \eqref{eq:dec-induction} holds for every child $\alpha'$ of $\alpha$.  Focus on any vertex $v$ involved in $\tau[\alpha]$.  If $v = \mu_\alpha$, then trivially there is a path from $\mu_\alpha$ to $v$ in $H_\alpha$. Otherwise, $v = \tail(e_\beta)$ for some leaf $\beta$ of $\tau[\alpha]$. Let $\alpha_2$ be the child of $\alpha$ such that $\beta \in V(\tau[\alpha_2])$.  By induction hypothesis, there is a path from $\mu_{\alpha_2}$ to $v$ in $H_{\alpha_2}$. If $\mu_{\alpha_2} = \mu_\alpha$, there is a path from $\mu_\alpha$ to $v$ in $H_{\alpha_2} \subseteq H_\alpha$. Otherwise, by Property~\ref{property:contain-root}, there is a child ${\alpha_1}$ of $\alpha$ such that $\mu_{\alpha_1} = \mu_\alpha$ and $\mu_{\alpha_2}$ is involved in $\tau[{\alpha_1}]$. By induction hypothesis, there is a path from $\mu_\alpha = \mu_{\alpha_1}$ to $\mu_{\alpha_2}$ in $H_{\alpha_1}$.  Since $H_\alpha$ contains both $H_{\alpha_2}$ and $H_{\alpha_1}$, there is a path from $\mu_\alpha$ to $v$ in $H_\alpha$. So, \eqref{eq:dec-induction} holds.
	\end{proof}

	\HugeTree*
	\begin{proof}
		The height of $T^0$ is easily seen to be $O(\bar h/g) = O(\log k/\log\log k)$. The number of children of a $p$-node is dominated by \fab{the} number of different twigs with a specific $\mu$ value for the root. This is at most $2^{2^g} \cdot n^{2\cdot 2^g} \leq n^{2^{g+2}} = n^{O(\log k)}$. \footnote{We first describe the shape of the tree. The perfect binary-tree of depth $g$ contains $2^g-1$ internal nodes we just need to specify whether each internal node has children or not. Now each node can have $n^2$ different choices for its $\mu$ and $e$ values.}  The number of children of a $q$-node is at most $2^g = O(\log k)$. Thus, the number of nodes in $T^0$ is at most $\left((\log k)n^{O(\log k)}\right)^{O(\log k/\log\log k)} = n^{O(\log^2k/\log\log k)}$.
	\end{proof}

	\lemmadecptoLCST*
	
	\begin{proof}

	We define a collapsing operation over a rooted tree $\bfT$ as follows. Given an internal node in $\bfT$ with exactly one child,  collapsing the node means removing the node and directly connect its child to its parent. If the node is the root of $\bfT$, we then simply remove the root.  It is easy to see that $\bfT^*$ satisfies the following properties:
	
	\begin{enumerate}[label = (A\arabic*), leftmargin=*]
		\item Every $p$-node in $\bfT^*$ has exactly one child which is a $q$-node.
		\item If a $q$-node is in $\bfT^*$, then all its children in $\bfT^0$ are in $\bfT^*$.
		\item Let $\tilde\bfT$ be the tree obtained from $\bfT^*$ by collapsing all $p$-nodes. Then $\tilde\bfT$ is \emph{isomorphic} to $\bfH$:   replacing each node $q$ in $\tilde\bfT$ with $\eta_q$ gives $\bfH$. \label{property:corresondence-bfT*-bfH}
%		\item Let $\tilde\bfT$ be the tree obtained from $\tilde \bfT'$ by removing $o$-nodes. Then $\tilde\bfT$ is \emph{isomorphic} to $\bfH$: replacing each node $q$ in $\tilde\bfT$ with $\eta_q$ gives $\bfH$. 
	\end{enumerate}

	With this correspondence, it is obvious that $\cost(\bfT^*) = \cost(\tau^*) = \opt$: this holds since the cost of a $q$ node is exactly the total cost of leaves of $\tau^*$ that are in $\eta_q$. We then show that $\bfT^*$ is indeed label-consistent. Notice that each $p$-node in $\bfT^*$ has exactly $1$ child in $\bfT^*$, and so the demand label for a $p$-node is satisfied. For a node $q$ in $\bfT^*$, all the demand labels added to $\dem(q)$ in Loop~\ref{State:clt-loop-labels-for-leaves} are satisfied since all children of $q$ are included in $\bfT^*$.  
	
	Now focus on a label $\ell'$ added to $\dem(q)$ in an iteration of Loop~\ref{State:clt-loop-labels-for-consistency}; let $\eta, \alpha, \alpha_1, {\alpha_2}$ be the values of the correspondent  variables in the end of the iteration. Since we assumed a label $\ell'$ was created and added to $\dem(q)$ in this iteration, we have $\mu_{\alpha_2} \neq \mu_\alpha$. As $\tau^*$ is a valid decomposition tree, there is a leaf $\beta'$ of $\tau^*[\alpha_1]$ such that $\tail(e_{\beta'}) = \mu_{\alpha_2}$.   If this leaf $\beta'$ is in $\eta$ then it is in $\eta[\alpha_1]$; in this case the label $\ell'$ can \fab{not} be created.  So, $\beta'$ is not in $\eta$, which means there is a leaf $\beta$ of $\eta[\alpha_1]$ with $e_\beta$ undefined, a twig $\eta' \in \bfH[\eta]$ with $\beta'$ being a leaf of $\eta'$.  By the correspondence between $\bfT^*$ and $\bfH$ in \ref{property:corresondence-bfT*-bfH}, there is a leaf $\beta$ in $\eta[\alpha_1]$ with $e_\beta$ undefined, and a node $q' \in  V(\bfT^q_\beta) \cap V(\bfT^*)$ such that $\eta_{q'}$ contains a leaf $\beta'$ with $e_{\beta'}$ defined and $\tail(e_{\beta'}) = \mu_{\alpha_2}$. Thus, the label $\ell'$ will be satisfied by this $q'$.  
	
	Finally, all the global demand labels $K$ are provided by $\bfT^*$: for every terminal $v$, $\tau^*$ contains a leaf $\beta$ with $\tail(e_\beta) = v$ this $\beta$ will appear in some twig $\eta$ and the node $q$ with $\eta_q = \eta$ will provide the label $v$. 
	\end{proof}

	\lemmaLCSTtodecompositiontree*
	
	\begin{proof}
		We pick the twigs $\eta_q$ over all nodes $q$ in $\bfT$.  For a technical issue, we also pick a singular \emph{root-twig} $\alpha$ with $\mu_\alpha = r$. 
		Then our decomposition tree $\tau$ is constructed by taking the  collection $\calC$ of twigs we picked, identifying some pairs of nodes in these twigs naturally.  We shall make sure that when we identify two nodes, they will have the same $\mu$-value and they do not have $e$ values.
		
		Focus on a node $p$ in $\bfT$.  Let $q$ be the parent of $p$ and $\beta$ be the leaf of $\eta_q$ such that $p = \root(\bfT^q_\beta)$; If $p = \root(\bfT^0)$, then $q$ is not defined and we let $\beta$ be the node in the root-twig.  %In the former case, we have $\mu_{\beta} = u_p$ and in the later case, we have $\mu_\beta = r  = u_p$. So $\mu_{\beta} = u_p$ always holds. 
		Then for every child $q'$ of $p$ in $\bfT$, we identify $\root(\eta_{q'})$ with $\beta$. Clearly we have $\mu_{\root(\eta_{q'})} = u_p = \mu_\beta$. $e_\beta$ is not defined since otherwise $\bfT^q_\beta$ does not exist. $e_{\root(\eta_{q'})}$ is not  defined either since $\eta_{q'}$ is a non-singular twig. 
		
		This finishes the construction of $\tau$.  We need to show that $\tau$ is a decomposition tree.  $\root(\tau)$ is the root of the root-twig and thus we have $\mu_{\root(\tau)} = r$.  For each leaf node $\beta$ with $e_\beta$ undefined in any twig $\eta_q$ in our collection $\calC$,  $p:=\root(\bfT^q_\beta)$ must be in $\bfT$ as all children of $q$ should be in $\bfT$ in order to make it label-consistent.   The label for $p$ must be satisfied by one of its children. Thus we must have identified $\beta$ with the root of some non-singular twig in $\calC$.  Thus $e$ values are defined for exactly the set of leaves of $\tau$.  Clearly, for a leaf $\beta'$ of $\tau$ we have $\mu_{\beta'} = \head(e_{\beta'})$; so Property~\ref{property:mu-is-head-of-e} also holds. 
		
		We then prove Property~\ref{property:contain-root} for an internal node $\alpha$ of $\tau$, and a child $\alpha_2$ of $\tau$ such that $\mu_{\alpha_2} \neq \mu_\alpha$.  The edge $(\alpha, \alpha_2)$ must be in $\eta_q$ for some $q$ in $\bfT$.  Let $\alpha_1$ be the other child of $\alpha$ in $\eta_q$. So, we have $\mu_{\alpha_1} = \mu_\alpha$ by the definition of a twig.  If there is a leaf $\beta$ of $\eta_q[\alpha_1]$ with $e_\beta$ defined and $\tail(e_\beta) = \mu_{\alpha_2}$, then Property~\ref{property:contain-root} holds for this $\alpha$ and $\alpha_2$.  Otherwise in the iteration of Loop~\ref{State:clt-loop-labels-for-consistency} in $\clt$ for this $\alpha, \alpha_1$ and $\alpha_2$, we have created a label $\ell'$. In order for this $\ell'$ to be satisfied, there must be a leaf $\beta$ of $\eta_q[\alpha_1]$ with $e_\beta$ undefined, and a twig $\eta' \in \bfH[\eta_q]$ that contains a leaf $\beta'$ with $e_{\beta'}$ defined and $\tail(e_{\beta'}) = \mu_{\alpha_2}$. Thus this $\beta'$ will be a leaf node of $\tau[\alpha_1]$; thus Property~\ref{property:contain-root} holds.
				
		The cost of $\bfT$ is exactly $\cost(\tau)$ since every $q$-node of $\bfT$ is correspondent to a twig $\eta_q$ with cost being the cost of leaves in $\eta_q$. 	If a global demand label $v \in K$ is provided by $\bfT^*$, then some node $q$ with $\eta_q$ containing a leaf $\beta$ with $\tail(e_\beta) = v$ will be in $\bfT$, and our $\tau$ will contain the leaf $\beta$ and thus $v$ will be involved in $\tau$.
	\end{proof}

\section{Hardness of DST for the Class of Quasi-Polynomial-Time Algorithms}
\label{appendix:hardness-quasi-poly}

In this section, we present the hardness result for the Directed Steiner Tree problem for the class of quasi-polynomial-time algorithms.

Our hardness result is a refinement of the hardness construction of Halperin and Krauthgamer \cite{HalperinK03}.
To avoid repeating all the proofs in \cite{HalperinK03}, it suffices for us to consider the size of the construction.
It is worth remarking that the hardness result of Halperin and Krauthgamer is designed for an instance of the {\em group Steiner tree problem} (GST) on a tree. To be formal, GST is defined as follows.

\begin{definition}
\label{def:group-steiner}
In GST, we are given an $n$-vertex {\bf undirected} graph $G$ with edge-costs, a root vertex $r$ and a collection of subsets $S_1,\ldots,S_k$ of vertices ({\em groups}); the goal is to find a minimum-cost subgraph that contains a path from the root to at least one vertex of each group.
\end{definition}

It can be seen that GST is a special case of DST. 
One can reduce GST to DST by first making $G$ bi-directed by making two copies for each of $G$, one for each direction, and then add a terminal $t_i$, for each group $S_i$, with zero cost edges directed from every vertex of $S_i$ to $t_i$.
Thus, we will focus on the construction and the size of the tree constructed in \cite{HalperinK03}.
The parameter that we are interested in is the number of  the group (as we claim the lower bound of $\Omega(\log^2k/\log\log{k})$). 

The starting point of the reduction is the {\em Label-Cover} problem defined below.

\begin{definition}[Label-Cover (a.k.a. Projection Game)]
\label{def:label-cover}
Let $G=(U,W; E)$ be a bipartite (directed) graph on $n$ vertices and $m$ edges. 
Let $\Sigma$ be a set of labels (or alphabet).
Each edge $(u,w)\in E$ (where $u\in U$ and $w\in W$) of the graph $G$ is associated with a projection $\pi_{uw}:\Sigma\rightarrow\Sigma$.
A labeling $f$ is an assignment $f:U\cup W\rightarrow\Sigma$ that assigns one label from $\Sigma$ to each vertex of $G$.
The labeling $f$ is said to {\em cover} an edge $(u,w)\in E(G)$ if $\pi_{uw}(f(u)) = f(w)$. 
The goal in the Label-Cover problem is to find a labeling that covers the maximum number of edges.
\end{definition}

To the best of our knowledge, the hardness factor $\log^{2-\epsilon)}k$, for any $\epsilon>0$, is the best one could prove under the standard assumption $\mathrm{NP}\not\subseteq\mathrm{ZPTIME}(n^{\mathrm{polylog}(n)})$.
To show a strong hardness, we need to assume a strongly: (1) The {\em Exponential-Time Hypothesis} (ETH) for $k$-SAT and (2) The {\em Projection Games Conjecture}.

\begin{hypothesis}[(randomized) Exponential-Time Hypothesis for $k$-SAT]
\label{hypo:ETH}
For any constant $k>0$, there exists a constant $c_k$ such that $k$-SAT admits no (randomized) $2^{c_k n}$-time algorithm.
In particular, there is no $2^{o(n)}$-time algorithm that solves SAT.
\end{hypothesis}

\begin{hypothesis}[Projection Games Conjecture \cite{Moshkovitz15}]
\label{hypo:projection-games}
There exists a constant $c > 0$ such that, for every $\epsilon>1/n^c$, a SAT instance $\phi$ on input of size $n$ can be efficiently reduced to a Label-Cover instance on a $\poly(1/\epsilon)$-regular bipartite graph with $n^{1+o(1)}$ veritces over a set of label of size $\poly(1/\epsilon)$ in such a way that
\begin{itemize}
\item {\bf Yes-Instance:} If $\phi$ is satisfiable, then there exists a labeling that covers all the edges of $G$.
\item {\bf No-Instance:} If $\phi$ is not satisfiable, then there exists no labeling that covers more than $\epsilon$ fraction of the edges of $G$.
\end{itemize}
\end{hypothesis}

Combining the two hypotheses, we may assert that, for any $1<\varepsilon\leq c$ (for $c$ from \Cref{hypo:projection-games}), there is no $2^{n^{\varepsilon}}$-time algorithm that approximates the Label-Cover problem to within a factor of $n^{\varepsilon}$. In fact, we do not need the full power of ETH and may weaken the assumption as below.

\begin{hypothesis}[ETH for Projection Games]
\label{hypo:ETH-PGC}
Unless $\mathrm{NP}\subseteq\cap_{\varepsilon>0}\mathrm{ZPTIME}(2^{n^{\varepsilon}})$,
there exists a constant $0 < \epsilon^*$ such that,
for any constant $0<\epsilon \leq \epsilon^*$, there exist constants $c_{\epsilon}, d_{\epsilon}$ and $\delta_{\epsilon}\leq\epsilon$ depending on $\epsilon$ such that, the Label-Cover problem on a $n^{d_\epsilon}$-regular bipartite graph $G$ with $n$ vertices and a set of labels of size $n^{c_{\epsilon}}$ admits no $2^{n^{\delta_{\epsilon}}}$-time algorithm that distinguishes the following two cases:
\begin{itemize}
	\item {\bf Yes-Instance:} There exists a labeling that covers all the edges of $G$.
	\item {\bf No-Instance:} There exists no labeling that covers more than $1/n^{\epsilon}$ fraction of the edges of $G$.
\end{itemize}
In particular, the Label-Cover problem admits no $2^{n^{\delta_{\epsilon}}}$-time $n^{\epsilon}$-approximation algorithm for any $0<\epsilon\leq\epsilon^*$ unless every NP problem admits a randomzied algorithm that runs in time $2^{n^{o(1)}}$.
\end{hypothesis}

Applying \Cref{hypo:ETH-PGC} to the proof in \cite{HalperinK03} with slightly different parameters setting immediately gives us the approximation hardness of $\Omega(\log^2 k/\log\log{k})$ for the directed Steiner tree problem. To avoid overwhelming readers with too much information (and avoid repeating the proof in \cite{HalperinK03}), we re-state the reduction in \cite{HalperinK03} as below.

\begin{theorem}[\cite{HalperinK03}]
\label{thm:HK-GST-reduction}
Consider an instance $\psi$ of the Label-Cover problem on a $\Delta$-regular $n$-vertex bipartite graph with a set of labels $\Sigma$ of size $\sigma$.
For any parameter $1 \leq h \leq O(\log^2 n)$, there exists a randomized reduction from $\psi$ to an instance of the Group Steiner Tree problem on a tree $T$ with costs on edges and with $k$ groups such that $|V(T)|=(\sigma n)^h$ and $k=\Delta n^h$. Moreover, with high probability, the following holds:
\begin{itemize}
\item {\bf Yes-Instance:} If there exists a labeling that covers all the edges of $G$, then there exists a feasible solution $T'\subseteq T$ to the Group Steiner Tree problem with $\cost(T')=h^2$.
\item {\bf No-Instance:} If there is no labeling that covers more than $\gamma$ fraction of the edges of $G$, then every feasible solution $T'\subseteq T$ to the Group Steiner Tree problem must have cost at least $\cost(T')\geq \min\{\gamma^{-1/2}h, \Omega(h\log {k})\}$.
\end{itemize}
\end{theorem}

We will now prove our hardness result, which can be considered as a corollary of \Cref{thm:HK-GST-reduction}.

\begin{theorem}
Suppose \Cref{hypo:ETH-PGC} is true, i.e., 
$\mathrm{NP}\not\subseteq\bigcap_{\varepsilon>0}\mathrm{ZPTIME}(2^{n^{\varepsilon}})$ and the Projection Games Conjectures holds.
Then there exists no quasi-polynomial-time algorithm for the Directed Steiner Tree problem on a graph with $N$ vertices (resp., the Group Steiner tree problem on a tree with $N$ vertices) that yields an approximation ratio of $o(\log^2k/\log\log{k})$ or $o(\log^2{N}/\log\log N)$.
\end{theorem}

\begin{proof}
Let $h$ be a parameter as in \Cref{thm:HK-GST-reduction} (which we will specify later).
We first take an instance of the Label-Cover problem on a bipartite graph $G=(U, W; E)$ on $n$ vertices and with the set of labels $\Sigma$ from \Cref{hypo:ETH-PGC}.
Thus, we have an instance of the Label-Cover problem with parameter $\Delta=n^{d_{\epsilon}}$ (which is the degree of $G$) and the set of labels of size $\sigma=n^{c_{\epsilon}}$, and $\gamma = {1/n^{\epsilon}}$. (The parameter $\gamma$ is usually called the soundness error in literature.)

Now we choose the parameter $h=n^z$, for some constant $0 < z < 1$, that will be specified later.
Observe that 
$$
k = \Delta n^{h}  
  \implies \log{k} = h\cdot\log(\Delta\cdot n) = h\cdot(\log{\Delta} + \log{n}) = h\cdot \Theta(\log{n})
  \implies h = \frac{\log{k}}{\Theta(\log{n})}
$$
Moreover, it is not hard to see that $\log{n} = \Theta(\log\log {k})$ because
$$
k=\Delta n^h=n^{h+d_{\epsilon}}=2^{(n^z+d_{\epsilon})\log_2{n}} \implies \log\log k = \Theta(\log{n})
$$
Therefore, we have the hardness gap of $\Omega(\log^2k/\log\log k)$ as claimed.
Observe that $\log |V(T)| = \Theta(\log k)$. Thus, we have the same hardness gap for both in terms of $k$ and that of $N=|V(T)|$.

Next we prove the running-time lower bound.
Assume for a contrary that there exists an algorithm for GST on the tree $T$ 
that runs in time $O(|V(T)|^{\log^\zeta{|V(T)|}})$, for some constant $\zeta>0$,
and yields approximation guarantee $o(\log^2k/\log\log{k})$.
Then by setting $z<\delta_{\epsilon}/(3\zeta)$,
we would have an algorithm that runs in time
\begin{align*}
O(|V(T)|^{\log^\zeta{|V(T)|}}) 
 &= O(((\sigma n)^h)^{(h\log(\sigma n))^{\zeta}})
 &&= O((\sigma n)^{(h^2\log(\sigma n))^{\zeta}})
 &&= 2^{(O(h^2\log^2(\sigma n)))^{\zeta}}\\
 &= 2^{(O(h\log(\sigma n))^{2\zeta}}
 &&= 2^{(O(n^{\delta_{\epsilon}/(3\zeta)}\log(n^{1+c_{\epsilon}}))^{2\zeta}}
 &&< 2^{n^{\delta_{\epsilon}}}
\end{align*}
This running time contradicts the statement of \Cref{hypo:ETH-PGC}.
\end{proof}
\section{Hardness of the Label-Consistent problem on General Graphs}
\label{appendix:hardness-lcsg}

This section provides the proof for the hardness of the generalization of the Label-Consistent Subtree problem to general graphs, which we may call the {\em Label-Consistent Subgraph} problem (LCSG).

\begin{definition}
\label{def:LCST-on-general-graphs}
In LCSG, the input is an undirected graph $G=(V,E)$ with vertex (or edge) costs with a root vertex $r$, a set of labels $L$ and a set of global labels $K\subseteq L$. Each vertex of $G$ is associated with a demand function $\dem(v)\subseteq L$ and a service function $\ser(v)\subseteq L$. The goal in LCSG is to find a minimum-cost subgraph $H\subseteq G$ such that the following two properties hold:
\begin{itemize}
	\item For every global label $t\in K$, there is a path from $r$ to $t$ in $H$.
	\item For every vertex $v\in V(H)$ and every label $\ell\in\dem(v)$, there exists a path from $v$ to a vertex $w$ with $\ell\in\ser(v)$.
\end{itemize}
\end{definition}

We will now show that LCSG is at least as hard as the minimum Label-Cover problem.
The definition of the minimum Label-Cover problem is slightly different from that of the (maximum) Label-Cover problem; here we are allowed to assign multiple labels to each vertex, but we have to cover all the edges. The formal definition of the minimum Label-Cover problem is defined as below.

\begin{definition}[Minimum Label-Cover (a.k.a. Min-Rep \cite{Kortsarz01}]
	\label{def:label-cover}
	Let $G=(U,W; E)$ be a bipartite (directed) graph on $n$ vertices and $m$ edges. 
	Let $\Sigma$ be a set of labels (or alphabet).
	Each edge $(u,w)\in E$ (where $u\in U$ and $w\in W$) of the graph $G$ is associated with a projection $\pi_{uw}:\Sigma\rightarrow\Sigma$.
	A multi-labeling $f$ is an assignment $f:U\cup W\rightarrow 2^{\Sigma}$ that assigns a set of labels from $\Sigma$ to each vertex of $G$.
	The multi-labeling $f$ is said to {\em cover} an edge $(u,w)\in E(G)$ if there exists a label $a\in f(u)$ and $b\in f(w)$ such that $\pi_{uw}(f(u)) = f(w)$. 
	The cost of the multi-labeling $f$ is $\sum_{v\in U\cup W}|f(v)|$.
	The goal in the Label-Cover problem is to find a multi-labeling with minimum-cost that covers all the edges.
\end{definition}

We remark that the standard hardness of the Label-Cover problem is between the case that the optimal solution is a labeling versus the case that the optimal solution is a multi-labeling.

\begin{theorem}[Hardness of the Label-Consistent problem on General Graphs]
There exists a polynomial-time reduction from an instance $\psi$ of the Label-Cover problem on a bipartite graph $G=(U,W; E)$ with $n$ vertices, $m$ edges and with the set of labels $\Sigma$ to an instance $I$ of the Label-Consistent Subgraph problem on a graph $G'$ on a set of label $L'$ of size $m+\Sigma$.
Moreover, the following holds. 
\begin{itemize}
	\item {\bf Yes-Instance:} Suppose there is a labeling $f$ that covers all the edges of $G$, then there exists a solution to the instance $I$ of LCSG with cost $n$.
	\item {\bf No-Instance:} Suppose there is no multi-labeling $f$ with cost $\gamma n$ that covers all the edges of $G$, then any feasible solution to the instance $I$ of LCSG must have cost at least $\gamma n$.
\end{itemize}
In particular, LCSG is at least as hard as the Label-Cover problem with perfect completeness.
\end{theorem}
\begin{proof}
First, take an instance of the Label-Cover problem consisting of a graph $G=(U,W; E)$ with the constraints $\pi_{uw}$ on edges $(u,w)\in E(G)$ and a set of label $L$.
We first construct a graph $G'$ by adding a root vertex $r$.
Then we add to $G'$ a set of vertices $U'=\{u_{a}: u \in U, a\in L\}$ and $W'=\{w_{b}: w\in W, b\in L\}$. 
For each vertex $u\in U$ (resp., $w\in W$), we denote by $U'(u)=\{u_a:a\in L\}$ (resp., $W'(w)=\{w_b:b\in L\}$)  the set of vertices corresponding to a vertex $u$ (resp., $w$) in $G$.

We add edges joining $r$ to every vertex of $U'$,
and we add an edge $u_aw_b$ to $G'$ if $(u,w)\in E(G)$ and $\pi_{uw}(a) = b$. 
We set cost of each vertex in $U'\cup W'$ to be one.

Now we define the set of labels of the instance $I$ of LCSG.
Let the set of all labels be $L=(U\cup W')\cup\Sigma$, and the set of global labels be $K=U$.
We assign the service label $\ser(u_a)=\{u\}$ and 
the demands $\dem(u_a)=\{w_b\in W': (u,w)\in E(G) \land \pi_{uw}(a)=b\}$, for all vertices $u_a\in U'$.
Next we assign the service labels $\dem(w_b)=\{w_b\}$, for all vertices $w_b\in W'$; these vertices have no demands (i.e., $\dem(w_b)=\emptyset$).
This completes the construction.

\medskip

\noindent{\bf Completeness.} 
Suppose there is a labeling $f$ that covers all the edges of $G$.
Then we choose the root vertex $r$, all the vertices $u_a\in U'$ such that $f(u)=a$, and all the vertices $w_b\in W'$ such that $f(w)=b$.
We denote such a subgraph by $H'\subseteq G'$.
By feasibility of $f$, we know that, for every vertex $u_a\in V(H')$ and for every edge $(u,w)\in E(G)$, there exists a vertex $w_b\in V(H')$ such that $\pi_{uw}(a)=b$; moreover, by construction, $H'$ must contain a path $(r,u_a,w_b)$. 
Consequently, for every global label $u\in K$, the graph $H'$ has an $r,u_a$-path, for $a=f(u)$,
and for every demand label $w_b\in\dem(u_a)$, we have a $u_a,w_b$-path (which is just a single edge).
Thus, the graph $H'$ is label-consistent and must be a feasible solution to the LCSG instance $I$ with the same cost as $f$.

\medskip

\noindent{\bf Soundness.} 
Suppose there is no multi-labeling of cost less than $\gamma n$ that covers all the edges of $G$.
Then we claim that every feasible solution to the instance $I$ of LCSG must have cost at least $\gamma n$.
Suppose to a contrary that there exists a subgraph $H'\subseteq G'$ that is feasible to the instance $I$ of LCSG, but $H'$ has cost less than $\gamma n$. 
Then we can obtain a feasible multi-labeling $f$ by assigning 
$f(u)=V(H')\cap U'(u)$, for all $u\in U$, and $f(w)=V(H')\cap W'(w)$, for all $w\in W$.
We know that, for every vertex $u_a\in V(H')$ and for all edges $(u,w)\in E(G)$, $H'$ must contain a $u_a,w_b$-path such that $b=f_{uw}(a)$.
This means that $a\in f(u)$, $b\in f(w)$ and that $f_{uw}(a)=b$, for every edge $(u,w)\in E(G)$, i.e., $f$ covers all the edges of $G$.
It is not hard to see that $f$ has the same cost as $H'$, i.e., $f$ has cost less than than $\gamma n$.
But, this is a contradiction since any multi-labeling that covers all the edges of $G$ must have cost at least $\gamma n$.
\end{proof}

\end{document}